\documentclass{scrartcl}

\usepackage{stmaryrd}
\usepackage{xspace}
\usepackage{color}
\usepackage{natbib}
\usepackage{amsthm, amssymb}
\usepackage{tikz}
\usetikzlibrary{matrix,arrows,shapes,calc}
\usepackage{subfig}

\newlength{\rulevgap}
\newlength{\ruleheight}
\newlength{\ruledepth}
\newsavebox{\rulebox}
\newlength{\GapLength}

\newcommand{\Rule}[2]{\savebox{\rulebox}[\width][b]                         %
                              {\( \frac{\raisebox{0in} {\( #1 \)}}       %
                                       {\raisebox{-0.03in}{\( #2 \)}} \)}   %
                      \settoheight{\ruleheight}{\usebox{\rulebox}}          %
                      \addtolength{\ruleheight}{\rulevgap}                  %
                      \settodepth{\ruledepth}{\usebox{\rulebox}}            %
                      \addtolength{\ruledepth}{\rulevgap}                   %
                      \raisebox{0in}[\ruleheight][\ruledepth]               %
                               {\usebox{\rulebox}}}
\newcommand{\Axiom}[1]{\savebox{\rulebox}[\width][b]                        %
                               {$\frac{}{\raisebox{-0.03in}{$#1$}}$}        %
                      \settoheight{\ruleheight}{\usebox{\rulebox}}          %
                      \addtolength{\ruleheight}{\rulevgap}                  %
                      \settodepth{\ruledepth}{\usebox{\rulebox}}            %
                      \addtolength{\ruledepth}{\rulevgap}                   %
                      \raisebox{0in}[\ruleheight][\ruledepth]               %
                               {\usebox{\rulebox}}}
\newcommand{\RuleSide}[3]{\Rule{#1}{#3}                          %
                          \:\mbox{\(#2\)}}

\usepackage{epigram}
\usepackage{epigram-desc}
\usepackage{epigram-ornament}
\usepackage{epigram-function}
\usepackage{epigram-inverse}
\usepackage{epigram-nat}
\usepackage{epigram-list}
\usepackage{epigram-fin}
\usepackage{epigram-vec}
\usepackage{epigram-enum}
\usepackage{epigram-sum}

\usepackage{wasysym}

\ColourEpigram

\usepackage{category}

\renewcommand{\IDesc}[1][\!]{\Canonical{Desc}\xspace\, #1}

\newcommand{\aContainer}{C}
\newcommand{\Shape}{S}
\newcommand{\Position}{P}
\newcommand{\Next}{n}

\newcommand{\MorphShape}{\sigma}
\newcommand{\MorphPosition}{\rho}
\newcommand{\MorphCoe}{q}

\newcommand{\aShape}{sh}
\newcommand{\aPosition}{pos}

\newcommand{\PosVars}{Xs}

\newcommand{\PolySource}{s}
\newcommand{\PolyTarget}{t}
\newcommand{\Polynome}{f}

\newcommand{\MorphSource}{u}
\newcommand{\MorphTarget}{v}
\newcommand{\MorphPull}{\alpha}
\newcommand{\MorphContra}{\omega}

\newcommand{\todo}[1]{\red{#1}}
\newcommand{\ie}{\emph{i.e.}\xspace}
\newcommand{\eg}{\emph{e.g.}\xspace}

\newtheorem{theorem}{Theorem}
\newtheorem{proposition}{Proposition}
\newtheorem{corollary}{Corollary}
\newtheorem{lemma}{Lemma}

\theoremstyle{definition}
\newtheorem{definition}{Definition}
\newtheorem{example}{Example}
\newtheorem{application}{Application}

\theoremstyle{remark}
\newtheorem{remark}{Remark}

\title{A Categorical Treatment of Ornaments\footnote{Revised \today}}
\author{Pierre-Evariste Dagand
        \and Conor McBride}
\date{}

\begin{document}

\maketitle

\begin{abstract}

Ornaments aim at taming the multiplication of special-purpose datatypes in dependently typed programming languages. In type theory, purpose is logic. By presenting datatypes as the combination of a structure and a logic, ornaments relate these special-purpose datatypes through their common structure. In the original presentation, the concept of ornament was introduced concretely for an example universe of inductive families in type theory, but it was clear that the notion was more general. This paper digs out the abstract notion of ornaments in the form of a categorical model. As a necessary first step, we abstract the universe of datatypes using the theory of polynomial functors. We are then able to characterise ornaments as cartesian morphisms between polynomial functors. We thus gain access to powerful mathematical tools that shall help us understand and develop ornaments. We shall also illustrate the adequacy of our model. Firstly, we rephrase the standard ornamental constructions into our framework. Thanks to its conciseness, this process gives us a deeper understanding of the structures at play. Secondly, we develop new ornamental constructions, by translating categorical structures into type theoretic artefacts. 

\end{abstract}


The theory of inductive types is generally understood as the study of
initial algebras in some appropriate category. A datatype definition
is therefore abstracted away as a signature functor that admits a
least fixpoint. This naturally leads to the study of polynomial
functors~\citep{gambino:poly-monads}, a class of functors that all
admit an initial algebra. These functors have been discovered and
studied under many guises. In type theory, they were introduced by
Martin-L\"{o}f under the name of \emph{well-founded
  trees}~\citep{martin-lof:itt,moerdijk:cat-w-trees,petersson:iw-types},
or W-types for short. Containers~\citep{abbott:containers} and their
indexed counterparts~\citep{morris:PhD} generalise these definitions to
a fibrational setting. Polynomial
functors~\citep{hyland:poly-fun,gambino:poly-monads} are the category
theorists' take on containers, working in a locally cartesian-closed
category.


There is a significant gap between this unified theoretical framework
and the implementations of inductive types: in systems such as
Coq~\citep{coq} or Agda~\citep{norell:agda}, datatypes are purely
syntactic artefacts. A piece of software, the positivity checker, is
responsible for checking that the definition entered by the user is
valid, \ie does not introduce a paradox. The power of the positivity
checker depends on the bravery of its implementers: for instance,
Coq's positivity checker is allegedly simple, therefore safer, but
rather restrictive. On the other hand, Agda's positivity checker is
more powerful, hence more complex, but also less trusted. For example,
the latter checks the positivity of functions in datatype
declarations, while the former conservatively rejects them. The more
powerful the positivity checker, the harder it is to relate the
datatype definitions to some functorial model.


An alternative presentation of inductive types is through a universe
construction~\citep{martin-lof:itt,dybjer:families,dagand:levitation}. The
idea is to reflect the grammar of polynomial functors into type theory
itself. Having internalised inductive types, we can formally
manipulate them and, for example, create new datatypes from old. The
notion of ornament~\citep{mcbride:ornament} is an illustration of this
approach. Ornaments arise from the realisation that inductive families
can be understood as the integration of a data-\emph{structure}
together with a data-\emph{logic}. The structure captures the dynamic,
operational behavior expected from the datatype. It corresponds to,
say, the choice between a list or a binary tree, which is governed by
performance considerations. The logic, on the other hand, dictates
the static invariants of the datatype. For example, by indexing lists
by their length, thus obtaining vectors, we integrate a logic of
length with the data. We can then take an \(m\times n\) matrix to be a
plainly rectangular \(m\)-vector of \(n\)-vectors, rather than a list
of lists together with a proof that measuring each length yields the
same result.

In dependent type theory, logic is purpose: when solving a problem, we
want to bake the problem's invariants into the datatype we
manipulate. Doing so, our code is correct by construction. The same
data-structure will be used for different purposes and will therefore
integrate as many logics: we assist to a multiplication of datatypes,
each built upon the same structure. This hinders any form of code
reuse and makes libraries next to pointless: every task requires us to
duplicate entire libraries for our special-purpose datatypes.


Ornaments tame this issue by organising datatypes along their
structure: given a datatype, an ornament gives an effective recipe to
extend -- introducing more information -- and refine -- providing a
more precise indexing -- the initial datatype. Applying that recipe
gives birth to a new datatype that \emph{shares the same structure} as
the original datatype. Hence, ornaments let us evolve datatypes with
some special-purpose logic without severing the structural ties
between them. In an earlier work\citep{dagand:fun-orn}, we have shown
how that information can be used to regain code reuse.


The initial presentation of ornaments and its subsequent
incarnation~\citep{mcbride:ornament,dagand:fun-orn} are however very
syntactic and tightly coupled with their respective universe of
datatypes. We are concerned that their syntactic nature obscures the
rather simple intuition governing these definitions. In this paper, we
give a semantic account of ornaments, thus exhibiting the underlying
structure of the original definitions. To do so, we adopt a
categorical approach and study ornaments in the framework of
polynomial functors. Our contributions are the
following:

\begin{itemize}


\item In Section~\ref{sec:desc-poly}, we formalise the connection
  between a universe-based presentation of datatypes and the theory of
  polynomial functors. In particular, we prove that the functors
  represented by our universe are equivalent to polynomial
  functors. This key result lets us move seamlessly from our concrete
  presentation of datatypes to the more abstract polynomial functors.


\item In Section~\ref{sec:orn}, we give a categorical presentation of
  ornaments as cartesian morphisms of polynomial functors. This
  equivalence sheds some light on the original definition of
  ornaments. It also connects ornaments to a mathematical object that
  has been widely studied: we can at last organise our universe of
  datatypes and ornaments on them into a category -- in fact a framed
  bicategory~\citep{shulman:framed-bicat} -- and start looking for
  categorical structures that would translate into interesting type
  theoretic objects.


\item In Section~\ref{sec:orn-structure}, we investigate the
  categorical structure of ornaments. The contribution here is
  twofold. On one hand, we translate the original, type theoretic
  constructions -- such as the ornamental algebra and the algebraic
  ornament -- in categorical terms and uncover the building blocks out
  of which they were carved out. On the other hand, we interpret the
  mathematical properties of ornaments into type theory -- such as the
  pullback of ornaments and the ornamentation of derivatives -- to
  discover meaningful software artefacts that were previously unknown.

\end{itemize}


Being at the interface between type theory and category theory, this
paper targets both communities. To the type theorist, we offer a more
semantic account of ornaments and use the intuition thus gained to
introduce new type theoretic constructions. Functional programmers
of the non-dependent kind will find a wealth of examples that should
help them grasp the more abstract concepts, both type theoretic and
category theoretic. To the category theorist,
we present a type theory, \ie a programming language, that offers an
interesting playground for categorical ideas. Our approach can be
summarised as \emph{categorically structured programming}. For
practical reasons, we do not work on categorical objects directly:
instead, we materialise these concepts through universes, thus
reifying categorical notions through computational objects. Ornaments
are merely an instance of that interplay between a categorical concept
-- cartesian morphism of polynomial functor -- and an effective, type
theoretic presentation -- the universe of ornaments. To help bridge
the gap between type theory and category theory, we have striven to
provide the type theorist with concrete examples of the categorical
notions and the category theorist with the computational intuition
behind the type theoretic objects.


\section{Categorical Toolkit}

In this section, we recall a few definitions and results from category
theory that will be used throughout this paper. None of these results
are new -- most of them are folklore -- we shall therefore not dwell
on the details. However, to help readers not familiar with these
tools, we shall give many examples, thus providing an intuition for
these concepts.


\subsection{Locally cartesian-closed categories}


Locally cartesian-closed categories (LCCC) were introduced by
\citet{seely:lccc} to give a categorical model of (extensional)
dependent type theory. A key idea of that presentation is the use of
adjunctions to model \(\Pi\)-types and \(\Sigma\)-types. 


\begin{definition}[Locally cartesian-closed category]

A locally cartesian-closed category is a category \(\LCCC{E}\) that is
pullback complete and such that, for \(\TypeAnn{f}{\Hom{\LCCC{E}}{X}{Y}}\), each
base change functor \(\TypeAnn{\Reindex{f}}{\Slice{\LCCC{E}}{Y} \to
  \Slice{\LCCC{E}}{X}}\), defined by pullback along \(f\), has a right
adjoint \(\CatForall{f}\).

\end{definition}

Throughout this paper, we work in a locally cartesian-closed category
\(\LCCC{E}\) with a terminal object \(\Terminal[\LCCC{E}]\) and sums. An object
\(\TypeAnn{f}{E \to I}\) in the slice \(\Slice{\LCCC{E}}{I}\) can be
thought of as an \(I\)-indexed set, which we shall informally denote
using a set comprehension notation
\(\ISet{E_i}{i \in I}\),
where \(E_i\) can be understood as the inverse image of \(f\) at
\(i\), or equivalently a fibre of the discrete fibration \(f\).

By construction, the base change functor has a left adjoint
\(\CatExists{f} = f \circ \_\). We therefore have the adjunctions
\[
    \CatExists{f} \leftAdj \Reindex{f} \leftAdj \CatForall{f}    
\]


Using a set theoretic notation, the base change functor writes as
reindexing by \(f\):
\[\Code{
\TypeAnn{\Reindex{f}}{\Slice{\LCCC{E}}{B} \to \Slice{\LCCC{E}}{A}} \\
\Reindex{f}\: \ISet{E_b}{b \in B} 
    \mapsto \ISet{E_{f a}}{a \in A}
}\]
While the left and right adjoints correspond to the following
definitions:
\[
\Code{
\TypeAnn{\CatExists{f}}{\Slice{\LCCC{E}}{A} \to \Slice{\LCCC{E}}{B}} \\
\CatExists{f}\: \ISet{E_a}{a \in A} 
    \mapsto \ISet{\sum_{a \in A_{b}} E_a}{b \in B}
}
\qquad
\Code{
\TypeAnn{\CatForall{f}}{\Slice{\LCCC{E}}{A} \to \Slice{\LCCC{E}}{B}} \\
    \CatForall{f}\: \ISet{E_a}{a \in A} \mapsto 
        \ISet{\prod_{a \in A_{b}} E_a}{b \in B}
}
\]
Where \(\sum\) and \(\prod\) respectively represent the (set
theoretic) disjoint union and cartesian product. Details of this construction
can be found elsewhere~\citep{maclane:sheaves}.

The internal language of \(\LCCC{E}\) corresponds to an extensional
type theory denoted \(\Set\), up to bureaucracy~\citep{curien:lccc-bureaucracy}. This
type theory comprises a unit type \(\Unit\), \(\Sigma\)-types,
\(\Pi\)-types, and equality is extensional. Syntactically, this type
theory is specified by the following judgments:
\[
\begin{array}{c@{\quad}c@{\quad}c}
\mbox{Formation rules:} 
&
\mbox{Introduction rules:}
&
\mbox{Elimination rules:}
\\
\Axiom{\TypeJudgment{\Gamma}{\Unit}{\Set}} 
&
\Axiom{\TypeJudgment{\Gamma}{\Void}{\Unit}} 
&
\\
\\
\Rule{\Code{
        \TypeJudgment{\Gamma}{\Meta{A}}{\Set}  \\
        \TypeJudgment{\Gamma}{\Meta{B}}{\Set}}}
     {\TypeJudgment{\Gamma}{\Meta{A} \Sum \Meta{B}}{\Set}}
&
\begin{array}[b]{@{}l}
\Rule{\TypeJudgment{\Gamma}{\Meta{a}}{\Meta{A}}}
     {\TypeJudgment{\Gamma}{\InjLeft[\Meta{a}]}{\Meta{A}}} \\
\Rule{\TypeJudgment{\Gamma}{\Meta{b}}{\Meta{B}}}
     {\TypeJudgment{\Gamma}{\InjRight[\Meta{b}]}{\Meta{B}}}
\end{array}
&
\Rule{\Code{
       \TypeJudgment{\Gamma}{\Meta{f}}{\Meta{A} \To \Meta{C}} \\
       \TypeJudgment{\Gamma}{\Meta{g}}{\Meta{B} \To \Meta{C}} \\
       \TypeJudgment{\Gamma}{\Meta{x}}{\Meta{A} \Sum \Meta{B}}}}
     {\TypeJudgment{\Gamma}{\SumSplit{\Meta{f}}{\Meta{g}}[\Meta{x}]}{\Meta{C}}}
\\
\\
\Rule{\Code{
       \TypeJudgment{\Gamma}{\Meta{S}}{\Set} \\
       \TypeJudgment{\Gamma ; \TypeAnn{x}{\Meta{S}}}{\Meta{T}}{\Set}}}
     {\TypeJudgment{\Gamma}{\SigmaTimes{\Var{x}}{\Meta{S}}{\Meta{T}}}{\Set}} 
&
\Rule{\Code{
        \TypeJudgment{\Gamma}{\Meta{a}}{\Meta{S}} \\
        \TypeJudgment{\Gamma}{\Meta{b}}{\Meta{T}[\Meta{a}/\Var{x}]}}}
     {\TypeJudgment{\Gamma}{\Pair{\Meta{a}}{\Meta{b}}}{\SigmaTimes{\Var{x}}{\Meta{S}}{\Meta{T}}}} 
&
\begin{array}[b]{@{}l}
\Rule{\TypeJudgment{\Gamma}{p}{\SigmaTimes{\X}{S}{T}}}
     {\TypeJudgment{\Gamma}{\Fst[p]}{S}}
\\
\Rule{\TypeJudgment{\Gamma}{p}{\SigmaTimes{\X}{S}{T}}}
     {\TypeJudgment{\Gamma}{\Snd[p]}{T[\Fst[p]/\X]}}
\end{array}
\\
\\
\Rule{\Code{
        \TypeJudgment{\Gamma}{S}{\Set}\\
        \TypeJudgment{\Gamma ; \XS}{T}{\Set}}}
     {\TypeJudgment{\Gamma}{\PiTo{\X}{S} T}{\Set}}
&
\Rule{\Code{
       \TypeJudgment{\Gamma}{S}{\Set} \\
       \TypeJudgment{\Gamma ; \XS}{t}{T}
     }}
     {\TypeJudgment{\Gamma}{\LamAnn{\X}{S} t}{\PiTo{\X}{S} T}}
&
\Rule{\Code{
       \TypeJudgment{\Gamma}{f}{\PiTo{\X}{S} T} \\
       \TypeJudgment{\Gamma}{s}{S}
      }}
     {\TypeJudgment{\Gamma}{f\: s}{T[s/\X]}}
\end{array}
\]
We chose to work in an extensional model for simplicity. However, all
the constructions presented in this paper have been modelled in Agda,
an intuitionistic type theory.


\subsection{Polynomials and polynomial functors}


Polynomials~\citep{hyland:poly-fun,gambino:poly-monads} provide a
categorical model for indexed families~\citep{dybjer:families} in a
LCCC. Polynomials themselves are small, diagrammatic objects that
admit a rich categorical structure. They are then \emph{interpreted}
as strong functors -- the polynomial functors -- between slices of
\(\LCCC{E}\). In this section, we shall illustrate the categorical
definitions with the corresponding notion on (indexed)
container~\citep{petersson:iw-types,hancock:is,morris:PhD}, an
incarnation of polynomials in the internal language \(\Set\).


\newcommand{\PolyDiag}[7]{#5 \stackrel{#4}{\longleftarrow} 
                          #2 \stackrel{#1}{\longrightarrow} #3
                          \stackrel{#6}{\longrightarrow} #7}

\begin{definition}[Polynomial~{\citep[\S{1.1}]{gambino:poly-monads}}]
  \label{def:poly}

A polynomial is the data of 3 morphisms \(\TypeAnn{\Polynome}{B \to
  A}\), \(\TypeAnn{\PolySource}{B \to I}\), and
\(\TypeAnn{\PolyTarget}{A \to J}\) in \(\LCCC{E}\). Conventionally, a
polynomial is diagrammatically represented by
\(
  \PolyDiag{\Polynome}{B}{A}
           {\PolySource}{I}
           {\PolyTarget}{J}
\).

\end{definition}


\newcommand{\IContainer}[3]{#1 {\mathop{\red{\lhd}}}^{#3} #2}
\Spacedcommand{\ICont}{\mathit{ICont}}
\newcommand{\IContainerMor}[3]{#1 \mathop{\red{\LHD}} #2}
\newcommand{\IContainerArr}[2]{\overset{#1}{\underset{#2}{\mathop{\blue{\Longrightarrow}}}}}

\begin{application}[Container]

In type theory, it is more convenient to work with (proof relevant)
predicates rather than arrows. Hence, inverting the arrow
\(\TypeAnn{\PolyTarget}{A \to J}\), we obtain a predicate
\(\TypeAnn{\Shape}{J \to \Set}\) -- called the
\emph{shapes}. Similarly, inverting \(\TypeAnn{\Polynome}{B \to A}\),
we obtain a predicate \(\TypeAnn{\Position}{\forall j.\: \Shape\: j
  \to \Set}\) -- called the \emph{positions}. The indexing map \(s\)
remains unchanged but, following conventional notation, we rename it
\(\Next\) -- the \emph{next index} function. We obtain the following
definition:
\[
\left\{
\begin{array}{@{}l@{\:}l}
\TypeAnn{\Var{\Shape}}{& \Var{J} \To \Set}\\
\TypeAnn{\Var{\Position}}{& 
                    \Var{\Shape}\:\Var{j} \To \Set}\\
\TypeAnn{\Var{\Next}}{& 
              \Var{\Position}\:\Var{\aShape} \To \Var{I}}
\end{array}
\right.
\]
Note that, to remove clutter, we (implicitly) universally quantify
unbound type variables, such as \(j\) in the definition of
\(\Position\) or \(\aShape\) in the definition of \(\Next\).  The data
of \(\Shape\), \(\Position\), and \(\Next\) is called a
\emph{container} and is denoted
\(\IContainer{\Meta{\Shape}}{\Meta{\Position}}{\Meta{\Next}}\). The
class of containers indexed by \(I\) and \(J\) is denoted
\(\ICont[\Meta{I}\: \Meta{J}]\). 

\end{application}

\begin{remark}[Intuition]

  Polynomials, and more directly containers, can be understood as
  multi-sorted signatures. The indices specifies the sorts. The shapes
  at a given index specify the set of symbols at that sort. The
  positions specify the arities of each symbol. The next index
  function specifies, for each symbol, the sort of its arguments.

\end{remark}


\begin{definition}[Polynomial functor~{\citep[\S 1.4]{gambino:poly-monads}}]
We \emph{interpret} a polynomial
\(\TypeAnn{F}{\PolyDiag{\Polynome}{B}{A}{\PolySource}{I}{\PolyTarget}{J}}\)
into a functor, conventionally denoted \(P_F\), between slices of
\(\LCCC{E}\) with the construction
\[
P_F \triangleq 
\Slice{\LCCC{E}}{I} \stackrel{\Reindex{\PolySource}}{\longrightarrow}
\Slice{\LCCC{E}}{B} \stackrel{\CatForall{\Polynome}}{\longrightarrow}
\Slice{\LCCC{E}}{A} \stackrel{\CatExists{\PolyTarget}}{\longrightarrow}
\Slice{\LCCC{E}}{J}
\]

A functor \(F\) is called \emph{polynomial} if it is isomorphic to the
interpretation of a polynomial, \ie there exists \(\PolySource\),
\(\Polynome\), and \(\PolyTarget\) such that \(F \Iso
\CatExists{\PolyTarget} \CatForall{\Polynome} \Reindex{\PolySource}\).

\end{definition}


\newcommandx{\InterpretIS}[2][2=\!]{\green{\llbracket} #1 \green{\rrbracket}_{\green{\CN{Cont}}} \xspace\:#2}

\begin{application}[Interpretation of container]

Unfolding this definition in the internal language, we interpret a
container as, first, a choice (\(\Sigma\)-type) of shape ; then, for
each (\(\Pi\)-type) position, a variable \(X\) taken at the next index
\(n\) for that position:
\[\Code{
\Let{\InterpretIS{\PiTel{\Var{\aContainer}}
                        {\ICont[\Var{I}\: \Var{J}]}} & 
                  \PiTel{\Var{X}}{\Var{I} \To \Set}}
    {\Var{J} \To \Set}
    {}\\
\Case{
\Return{\InterpretIS{\IContainer{\Var{\Shape}}
                                {\Var{\Position}}
                                {\Var{\Next}}}\:
        \Var{X}}
       {\Lam{\Var{j}}
        \SigmaTimes{\Var{\aShape}}{\Var{\Shape}\: \Var{j}}
                   {(\PiTo{\Var{\aPosition}}
                          {\Var{\Position}\: \Var{\aShape}}
                          {\Var{X}\: (\Var{\Next}\: \Var{\aPosition})})}}
}}
\]
hence justifying the name \emph{polynomial functor}: a polynomial
interprets into an \(\Shape\)-indexed sum of monomials \(X\) taken at
some exponent \(\TypeAnn{\aPosition}{\Position\: \aShape}\), or put
informally:
\[
\InterpretIS{\IContainer{\Var{\Shape}}{\Var{\Position}}{\Var{\Next}}}
            [\ISet{X_i}{i \in I}]
    \mapsto
        \ISet{\sum_{\aShape \in \Var{\Shape}_{j}} 
              \prod_{\aPosition \in \Var{\Position}_{\aShape}} 
              X_{\Var{\Next}\: \aPosition}}
             {j \in J}
\]

\end{application}


\Spacedcommand{\NatICont}{\Function{NatCont}}
\Spacedcommand{\NatShape}{\Function{\Shape}_{\green{\CN{Nat}}}}
\Spacedcommand{\NatPos}{\Function{\Position}_{\green{\CN{Nat}}}}
\Spacedcommand{\NatNext}{\Function{\Next}_{\green{\CN{Nat}}}}

\Spacedcommand{\ListICont}{\Function{ListCont}}
\Spacedcommand{\ListShape}{\Function{\Shape}_{\green{\CN{List}}}}
\Spacedcommand{\ListPos}{\Function{\Position}_{\green{\CN{List}}}}
\Spacedcommand{\ListNext}{\Function{\Next}_{\green{\CN{List}}}}

\Spacedcommand{\VecICont}{\Function{VecCont}}
\Spacedcommand{\VecShape}{\Function{\Shape}_{\green{\CN{Vec}}}}
\Spacedcommand{\VecPos}{\Function{\Position}_{\green{\CN{Vec}}}}
\Spacedcommand{\VecNext}{\Function{\Next}_{\green{\CN{Vec}}}}

\begin{figure}[!t]

\centering
\subfloat[][Natural number]{
\footnotesize
\(\Code{
  \NatICont \triangleq \medskip\\
    \:\:
    \left\{
    \Code{
    \Let{\NatShape & \PiTel{\Void}{\Unit}}{\Set}{
      \Return{\NatShape & \Void}{\Unit \Sum \Unit}} 
    \medskip\\
    \Let{\NatPos & \PiTel{\Var{\aShape}}{\NatShape[\Void]}}{\Set}{
      \Return{\NatPos & (\InjLeft[\Void])}{\Empty}
      \Return{\NatPos & (\InjRight[\Void])}{\Unit}} 
    \medskip\\
    \Let{\NatNext & \PiTel{\Var{\aPosition}}{\NatPos[\Var{\aShape}]}}
        {\Unit}{
      \Return{\NatNext & \Var{\aPosition}}{\Void}}
    }
    \right.
}\) \label{fig:nat-cont}
}
\subfloat[][List]{
\footnotesize
\(\Code{
\ListICont_{\Meta{A}} \triangleq \medskip \\
    \:\:
    \left\{
    \Code{
      \Let{\ListShape & \PiTel{\Void}{\Unit}}{\Set}{
      \Return{\ListShape & \Void}{\Unit \Sum \Meta{A}}} 
      \medskip \\
      \Let{\ListPos & \PiTel{\Var{\aShape}}{\ListShape[\Void]}}{\Set}{
      \Return{\ListPos & (\InjLeft[\Void])}{\Empty}
      \Return{\ListPos & (\InjRight[\Var{a}])}{\Unit}} 
      \medskip \\
      \Let{\ListNext & \PiTel{\Var{\aPosition}}{\ListPos[\Var{\aShape}]}}
          {\Unit}{
      \Return{\ListNext & \Var{\aPosition}}{\Void}}
    }
    \right.
}\) \label{fig:list-cont}
} 
\subfloat[][Vector]{
\footnotesize
\(\Code{
\VecICont_{\Meta{A}} \triangleq \medskip\\
    \:\:
    \left\{
    \Code{
      \Let{\VecShape & \PiTel{\Var{n}}{\Nat}}{\Set}{
        \Return{\VecShape & \Zero}{\Unit}
        \Return{\VecShape & (\Suc[\Var{n}])}{\Meta{A}}} 
      \medskip \\
      \Let{\VecPos & \PiTel{\Var{n}}{\Nat} 
                   & \PiTel{\Var{\aShape}}{\VecShape[\Var{n}]}}
          {\Set}{
        \Return{\VecPos & \Zero & \Void}{\Empty}
        \Return{\VecPos & (\Suc[\Var{n}]) & \Var{a}}{\Unit}} 
      \medskip \\
      \Let{\VecNext & \PiTel{\Var{n}}{\Nat} 
                    & \PiTel{\Var{\aShape}}{\VecShape[\Var{n}]}
                    & \PiTel{\Var{\aPosition}}{\VecPos[\Var{\aPosition}]}}
          {\Nat}{
      \Return{\VecNext & (\Suc[\Var{n}]) & \Var{a} & \Void}{\Var{n}}
      }
    }
    \right.
}\) \label{fig:vec-cont}
}

\caption{Examples of containers}

\end{figure}


\begin{example}[Container: natural number]
  \label{example:NatICont}

Natural numbers are described by the signature functor \(X \mapsto 1 +
X\). The corresponding container is given in
\figurename~\ref{fig:nat-cont}. There are two shapes, one to represent the
\(\Zero\) case, the other to represent the successor case,
\(\Suc\). For the positions, none is offered by the \(\Zero\) shape,
while the \(\Suc\) shape offers one. Note that the signature functor
is not indexed: the container is therefore indexed by the unit set and
the next index is trivial.

\end{example}


\begin{example}[Container: list]
  \label{example:ListICont}

The signature functor describing a list of parameter \(A\) is \(X
\mapsto 1 + A \times X\). The container is presented
\figurename~\ref{fig:list-cont}. Note the similarity with natural
numbers. There are \(1 + A\) shapes, \ie either the empty list
\(\Nil\) or the list constructor \(\Cons\) of some
\(\TypeAnn{a}{A}\). There are no subsequent position for the \(\Nil\)
shape, while one position is offered by the \(\Cons\) shapes. Indices
are trivial, for lists are not indexed.

\end{example}


\begin{example}[Container: vector]
  \label{example:VecICont}

To give an example of an indexed datatype, we consider vectors, \ie
lists indexed by their length. The signature functor of vectors is
given by
\(\ISet{X_n}{n \in \Nat} \mapsto 
    \ISet{n \PropEqual \Zero}{n \in \Nat} +
    \ISet{A \times X_{n-1}}{n \in \Nat^*}
\)
where the empty vector \(\Nil\) requires the length \(n\) to be
\(\Zero\), while the vector constructor \(\Cons\) must have a length
\(n\) of at least one and takes its recursive argument \(X\) at index
\(n - 1\).
The container representing this signature is given
\figurename~\ref{fig:vec-cont}. At index \(\Zero\), only the \(\Nil\) shape
is available while index \(\Suc[\Var{n}]\) offers a choice of
\(\TypeAnn{a}{\Meta{A}}\) shapes. As for lists, the \(\Nil\) shape has
no subsequent position while the \(\VCons\) shapes offer one. It is
necessary to compute the next index (\ie the length of the tail) only
when the input index is \(\Suc[n]\), in which case the next index is
\(n\).

\end{example}

We leave it to the reader to verify that the interpretation of
\(\NatICont\) (Example~\ref{example:NatICont}), \(\ListICont\)
(Example~\ref{example:ListICont}), and \(\VecICont\)
(Example~\ref{example:VecICont}) are indeed equivalent to the
signature functors we aimed at representing. With this exercise, one
gains a better intuition of the respective contribution of shapes,
positions, and the next index to the encoding of signature functors.


\begin{definition}[Polynomial morphism~{\citep[\S{3.8}]{gambino:poly-monads}}]
  \label{def:poly-morphism}

A morphism from
\(\TypeAnn{F}{\PolyDiag{\Polynome'}{B}{A}{\PolySource'}{I}{\PolyTarget'}{J}}\)
to
\(\TypeAnn{G}{\PolyDiag{\Polynome}{D}{C}{\PolySource}{K}{\PolyTarget}{L}}\)
is uniquely represented -- up to the choice of pullback -- by the
diagram:
\[
\begin{tikzpicture}
\matrix (m) [matrix of math nodes
            , row sep=2em
            , column sep=5em
            , text height=1.5ex
            , text depth=0.25ex
            , ampersand replacement=\&]
{ 
    I  \& B  \& A  \& J \\
       \& D' \& A  \&    \\  
    K  \& D  \& C  \& L \\
};
\path[->] 
   (m-1-2) edge node[above] {\(\PolySource'\)} (m-1-1)
   (m-1-2) edge node[above] {\(\Polynome'\)} (m-1-3)
   (m-1-3) edge node[above] {\(\PolyTarget'\)} (m-1-4)
   (m-3-2) edge node[below] {\(\PolySource\)} (m-3-1)
   (m-3-2) edge node[below] {\(\Polynome\)} (m-3-3)
   (m-3-3) edge node[below] {\(\PolyTarget\)} (m-3-4)
   (m-1-1) edge node[left] { \(\MorphSource\) } (m-3-1)
   (m-1-4) edge node[right] { \(\MorphTarget\) } (m-3-4)
   (m-2-2) edge (m-3-2)
   (m-2-2) edge (m-2-3)
   (m-2-3) edge node[right] { \(\MorphPull\) } (m-3-3)
   (m-2-2) edge node[left] { \(\MorphContra\) } (m-1-2)
   (m-1-3) edge[double,-] (m-2-3)
;
\begin{scope}[shift=($(m-2-2)!.25!(m-3-3)$)]
\draw +(-.25,0) -- +(0,0)  -- +(0,.25);
\end{scope}
\end{tikzpicture}
\]

\end{definition}


\begin{example}[Container morphism]

Let \(\TypeAnn{\Var{\MorphSource}}{\Var{I} \To \Var{K}}\) and
\(\TypeAnn{\Var{\MorphTarget}}{\Var{J} \To \Var{L}}\). A morphism from
a container
\(\IContainer{\Var{\Shape'}}{\Var{\Position'}}{\Var{\Next'}}\) to a
container \(\IContainer{\Var{\Shape}}{\Var{\Position}}{\Var{\Next}}\)
framed by \(\MorphSource\) and \(\MorphTarget\) is given by two
functions and a coherence condition:
\[
\left\{
\begin{array}{@{}l@{\:}l}
\TypeAnn{\Var{\MorphShape}}
        {& 
          \Var{\Shape'}\: \Var{j} \To 
          \Var{\Shape}\: (\Var{\MorphTarget}\: \Var{j})} \\
\TypeAnn{\Var{\MorphPosition}}
        {& 
           \Var{\Position}\: (\Var{\MorphShape}\: \Var{\aShape'}) \To 
           \Var{\Position'}\: \Var{\aShape'}} \\
\TypeAnn{\Var{\MorphCoe}}
        {& \Forall{\Var{\aShape'}}
                  {\Var{\Shape'}\: \Var{j}}
           \Forall{\Var{\aPosition}}
                  {\Var{\Position}\: (\Var{\MorphShape}\: \Var{\aShape'})}
           \Var{\MorphSource}\: (\Var{\Next'}\: (\Var{\MorphPosition}\: \Var{\aPosition}))
           \PropEqual
           \Var{\Next}\: \Var{\aPosition}}
\end{array}
\right.
\]
Remark that \(\MorphShape\) and \(\MorphPosition\) correspond exactly
to their diagrammatic counterparts, respectively \(\MorphPull\) and
\(\MorphContra\), while the coherence condition \(\MorphCoe\) captures
the commutativity of the left square. Commutativity of the right
square is ensured by construction, since we reindex \(\Shape\) by
\(\MorphTarget\) in the definition of \(\MorphShape\).

A container morphism, \ie the data \(\MorphShape\),
\(\MorphPosition\), and \(\MorphCoe\), is denoted
\(\IContainerMor{\Meta{\MorphShape}}{\Meta{\MorphPosition}}{\Meta{\MorphCoe}}\)
(leaving implicit the coherence condition). The hom-set of morphisms
from \(\IContainer{\Var{\Shape'}}{\Var{\Position'}}{\Var{\Next'}}\) to
\(\IContainer{\Var{\Shape}}{\Var{\Position}}{\Var{\Next}}\) framed by
\(\MorphSource\) and \(\MorphTarget\) is denoted
\(\IContainer{\Var{\Shape'}}{\Var{\Position'}}{\Var{\Next'}}
\IContainerArr{\Var{\MorphSource}}{\Var{\MorphTarget}}
\IContainer{\Var{\Shape}}{\Var{\Position}}{\Var{\Next}}\).

\end{example}


In this paper, we are particularly interested in a sub-class of
polynomial morphisms: the class of cartesian morphisms. Cartesian
morphisms represent only cartesian natural
transformations -- \ie for which the naturality
square forms a pullback.
\begin{definition}[Cartesian morphism~{\citep[\S 3.14]{gambino:poly-monads}}]
  \label{def:poly-cart-morphism}

A cartesian morphism from
\(\TypeAnn{F}{\PolyDiag{\Polynome'}{B}{A}{\PolySource'}{I}{\PolyTarget'}{J}}\)
to
\(\TypeAnn{G}{\PolyDiag{\Polynome}{D}{C}{\PolySource}{K}{\PolyTarget}{L}}\)
is uniquely represented by the diagram:
\[
\begin{tikzpicture}
\matrix (m) [matrix of math nodes
            , row sep=3em
            , column sep=5em
            , text height=1.5ex
            , text depth=0.25ex
            , ampersand replacement=\&]
{ 
    I  \& B  \& A  \& J \\
    K  \& D  \& C  \& L \\
};
\path[->] 
   (m-1-2) edge (m-1-1)
   (m-1-2) edge (m-1-3)
   (m-1-3) edge (m-1-4)
   (m-2-2) edge (m-2-1)
   (m-2-2) edge (m-2-3)
   (m-2-3) edge (m-2-4)
   (m-1-1) edge node[left] { \(\MorphSource\) } (m-2-1)
   (m-1-4) edge node[right] { \(\MorphTarget\) } (m-2-4)
   (m-1-2) edge (m-2-2)
   (m-1-2) edge (m-1-3)
   (m-1-3) edge node[right] { \(\MorphPull\) } (m-2-3)
;
\begin{scope}[shift=($(m-1-2)!.25!(m-2-3)$)]
\draw +(-.25,0) -- +(0,0)  -- +(0,.25);
\end{scope}
\end{tikzpicture}
\]
Where the \(\MorphPull\) is pulled back along \(\Polynome\), as
conventionally indicated by the right angle symbol.

\end{definition}


\newcommand{\IContainerCartMor}[3]{#1 \mathop{\red{\LHD^{c}}}}
\newcommand{\IContainerCartArr}[2]{\overset{#1}{\underset{#2}{\mathop{\blue{\Longrightarrow^c}}}}}

\begin{application}[Cartesian morphism of containers]

In the internal language, a cartesian morphism from
\(\IContainer{\Meta{\Shape'}}{\Meta{\Position'}}{\Meta{\Next'}}\) to
\(\IContainer{\Meta{\Shape}}{\Meta{\Position}}{\Meta{\Next}}\)
framed by \(\MorphSource\) and \(\MorphTarget\) corresponds to the
triple:
\[
\left\{
\begin{array}{@{}l@{\:}l}
\TypeAnn{\Var{\MorphShape}}{& \Var{\Shape'}\: \Var{j} \To \Var{\Shape}\: (\Var{\MorphTarget}\: \Var{j})} \\
\TypeAnn{\Var{\MorphPosition}}{& \Forall{\Var{\aShape'}}{\Var{\Shape'}\: \Var{j}} 
                        \Var{\Position}\: (\Var{\MorphShape}\: \Var{\aShape'}) \PropEqual \Var{\Position'}\: \Var{\aShape'}} \\
\TypeAnn{\Var{\MorphCoe}}{& \Forall{\Var{\aShape'}}{\Var{\Shape'}\: \Var{j}} 
                    \Forall{\Var{\aPosition}}{\Var{\Position}\: (\Var{\MorphShape}\: \Var{\aShape'})}
                         \Var{\MorphSource}\: (\Var{\Next'}\: \Var{\aPosition}) \PropEqual \Var{\Next}\: \Var{\aPosition}}
\end{array}
\right.
\]
The diagrammatic morphism \(\MorphPull\) translates into an operation
on shapes, denoted \(\MorphShape\). The pullback condition
translates into a proof \(\MorphPosition\) that the source positions
are indeed obtained by pulling back the target positions along
\(\MorphShape\). As for the indices, the coherence condition
\(\MorphCoe\) captures the commutativity of the left
square. Commutativity of the right square is ensured by construction,
since we reindex \(\Shape\) by \(\MorphTarget\) in the definition of
\(\MorphShape\).

A cartesian morphism is denoted
\(\IContainerCartMor{\Var{\MorphShape}}{\Var{\MorphPosition}}{\Var{\MorphCoe}}\),
leaving implicit the proof obligations. The hom-set of cartesian
morphisms from
\(\IContainer{\Var{\Shape'}}{\Var{\Position'}}{\Var{\Next'}}\) to
\(\IContainer{\Var{\Shape}}{\Var{\Position}}{\Var{\Next}}\) is
denoted \(\IContainer{\Var{\Shape'}}{\Var{\Position'}}{\Var{\Next'}}
\IContainerCartArr{\Var{\MorphSource}}{\Var{\MorphTarget}}
\IContainer{\Var{\Shape}}{\Var{\Position}}{\Var{\Next}}\). Because
polynomials and containers conventionally use different notations, we
sum-up the equivalences in \figurename~\ref{fig:poly-cont}.

\end{application}


\begin{example}[Cartesian morphism]
  \label{example:cart-morph-list-nat}

We build a cartesian morphism from \(\ListICont_{\Meta{A}}\)
(Example~\ref{example:ListICont}) to \(\NatICont\)
(Example~\ref{example:NatICont}) by mapping shapes of
\(\ListICont_{\Meta{A}}\) to shapes of \(\NatICont\):
\[
\Let{\green{\MorphShape} & 
         \PiTel{\Var{\aShape_l}}
               {\ListShape[\Void]}}
    {\NatShape[\Void]}{
  \Return{\green{\MorphShape} & (\InjLeft[\Void])}
         {\InjLeft[\Void] \qquad \CommentLine{\(\Nil\) to \(\Zero\)}}
  \Return{\green{\MorphShape} & (\InjRight[\Var{a}])}
         {\InjRight[\Void] \qquad \CommentLine{\(\Cons[\Var{a}]\) to \(\Suc\)}}
}
\]
We are then left to check that positions are isomorphic: this is
indeed true, since, in the \(\Nil\)/\(\Zero\) case, there is no
position while, in the \(\Cons\)/\(\Suc\) case, there is only one
position. The coherence condition is trivially satisfied, since both
containers are indexed by \(\Unit\). We shall relate this natural
transformation to the function computing the length of a list in
Example~\ref{ex:orn-alg-list}.

\end{example}


\begin{figure}[t]

  \centering
  \begin{tabular}{|c|c|l|}
    \hline
    Polynomial & Container & Obtained by\\
    \hline
    \(\TypeAnn{\PolyTarget}{A \to J}\) & 
    \(\TypeAnn{\Shape}{\Var{J} \To \Set}\) &
    Inverse image 
    \\
    \(\TypeAnn{\Polynome}{B \to A}\) & 
    \(\TypeAnn{\Position}{\Var{\Shape}\:\Var{j} \To \Set}\) &
    Inverse image
    \\
    \(\TypeAnn{\PolySource}{B \to I}\) &
    \(\TypeAnn{\Var{\Next}}{\Var{\Position}\:\Var{\aShape} \To \Var{I}}\) &
    Identity
    \\
    \(\TypeAnn{\MorphPull}{A \to C}\) & 
    \(\TypeAnn{\Var{\MorphShape}}{\Var{\Shape'}\: \Var{j} \To \Var{\Shape}\: (\Var{\MorphTarget}\: \Var{j})}\) &
    Identity \\
    \hline
  \end{tabular}
  
  \caption{Translation polynomial/container}
  \label{fig:poly-cont}

\end{figure}


We have seen that polynomials interpret to (polynomial) functors. We
therefore expect morphisms of polynomials to interpret to natural
transformations between these functors.

\begin{definition}[Interpretation of polynomial morphism~{\citep[\S{2.1} and \S{2.7}]{gambino:poly-monads}}]

For the morphism of polynomials given in
Definition~\ref{def:poly-morphism}, we construct the following natural
transformation:
\[
\begin{tikzpicture}
\matrix (m) [matrix of math nodes
            , row sep=3em
            , column sep=5em
            , text height=1.5ex
            , text depth=0.25ex
            , ampersand replacement=\&]
{ 
    \SliceL{E}{I}  \& 
        \SliceL{E}{B} \&[-2.5em] 
        \SliceL{E}{B} \&[-2.5em] 
        \SliceL{E}{A} \&[-2.5em] 
                      \&[-2.5em] 
        \SliceL{E}{J} \\
    \SliceL{E}{K} \&
        \SliceL{E}{D'} \& 
                       \&
        \SliceL{E}{A} \&
                      \&
        \SliceL{E}{J} \\
    \SliceL{E}{K} \&
        \SliceL{E}{D} \&
                      \&
        \SliceL{E}{C} \&
        \SliceL{E}{C} \&
        \SliceL{E}{L} \\
};
\node at ($(m-1-1) !.5! (m-2-2)$) {\(\Iso\)};
\node at ($(m-1-3) !.5! (m-2-4)$) {\(\Iso\)};
\node at ($(m-1-4) !.5! (m-2-6)$) {\(\Iso\)};
\node at ($(m-2-1) !.5! (m-3-2)$) {\(\Iso\)};
\node at ($(m-2-2) !.5! (m-3-4)$) {\(\Iso\)};
\node at ($(m-3-5) !.5! (m-2-6)$) {\(\Iso\)};
\node at (barycentric cs:m-1-2=.3,m-1-3=.3,m-2-2=.3) {\(\Downarrow\eta\)};
\node at (barycentric cs:m-2-4=.3,m-3-4=.3,m-3-5=.3) {\(\Downarrow\epsilon\)};
\path[->] 
    (m-1-1) edge node[above] {\(\Reindex{\PolySource'}\)} (m-1-2)
    (m-1-2) edge[-,double] (m-1-3)
    (m-1-3) edge node[above] {\(\CatForall{\Polynome'}\)} (m-1-4)
    (m-1-4) edge node[above] {\(\CatExists{\PolyTarget'}\)} (m-1-6)
    (m-3-1) edge node[below] {\(\Reindex{\PolySource}\)} (m-3-2)
    (m-3-2) edge node[below] {\(\CatForall{\Polynome}\)} (m-3-4)
    (m-3-4) edge[-,double] (m-3-5)
    (m-3-5) edge node[below] {\(\CatExists{\PolyTarget}\)} (m-3-6)
   (m-2-1) edge node[above] {\(\Reindex{\PolySource \circ \MorphPull^\dagger}\)} (m-2-2)
   (m-2-2) edge node[below] {\(\CatForall{{\Polynome}^\dagger}\)} (m-2-4)
   (m-2-4) edge node[above] {\(\CatExists{\PolyTarget}\)} (m-2-6)
   (m-1-1) edge node[left] { \(\CatExists{\MorphSource}\) } (m-2-1)
   (m-2-1) edge[-,double] (m-3-1)
   (m-1-6) edge[-,double] (m-2-6)
   (m-2-6) edge node[right] { \(\CatExists{\MorphTarget}\) } (m-3-6)
   (m-3-2) edge node[right] {\(\Reindex{\MorphPull^\dagger}\)} (m-2-2)
   (m-3-4) edge node[left] {\(\Reindex{\MorphPull}\)} (m-2-4)
   (m-2-4) edge node[above right] {\(\CatExists{\MorphPull}\)} (m-3-5)
   (m-1-2) edge node[left] {\(\Reindex{\MorphContra}\)} (m-2-2)
   (m-2-2) edge node[below right] {\(\CatForall{\MorphContra}\)} (m-1-3)
   (m-1-4) edge[-,double] (m-2-4)
;
\end{tikzpicture}
\]

\end{definition}



The diagrammatic construction of the interpretation on morphism is
perhaps intimidating. Its actual simplicity is revealed by containers,
in the internal language.

\begin{example}[Interpretation of container morphism]

A morphism from
\(\IContainer{\Var{\Shape'}}{\Var{\Position'}}{\Var{\Next'}}\) to
\(\IContainer{\Var{\Shape}}{\Var{\Position}}{\Var{\Next}}\) simply
maps shapes \(\Shape'\) to shapes \(\Shape\) covariantly using
\(\MorphShape\) and maps positions \(\Position\) to positions
\(\Position'\) contravariantly using \(\MorphPosition\):
\[\Code{
\Let{\InterpretIS{\PiTel{\Var{m}}
                        {\IContainer{\Var{\Shape'}}{\Var{\Position'}}{\Var{\Next'}}
                          \IContainerArr{\Var{\MorphSource}}{\Var{\MorphTarget}}
                          \IContainer{\Var{\Shape}}{\Var{\Position}}{\Var{\Next}}}} & 
                  \PiTel{\Var{xs}}
                        {\InterpretIS{\IContainer{\Var{\Shape'}}
                                                 {\Var{\Position'}}
                                                 {\Var{\Next'}}}[\Var{X}]}}
    {\InterpretIS{\IContainer{\Var{\Shape}}
                             {\Var{\Position}}
                             {\Var{\Next}}}[\Var{X}]}
    {}\\
\Case{
\Return{\InterpretIS{\IContainerMor{\Var{\MorphShape}}
                                   {\Var{\MorphPosition}}
                                   {\Var{\MorphCoe}}} &
        \Pair{\Var{\aShape'}}{\Var{\PosVars}}}
       {\Pair{\Var{\MorphShape}\: \Var{\aShape'}}
             {\Var{\PosVars} \Compose \Var{\MorphPosition}}}
}
}
\]
Note that, thinking of \(X\) as being parametrically quantified, there
is not much choice anyway: the shapes are in a covariant position while
positions are on the left on an arrow, \ie contravariant position.

\end{example}


\citet{hancock:is} present these morphisms as defining a
(constructive) simulation relation: having a morphism from \(C'
\triangleq
\IContainer{\Var{\Shape'}}{\Var{\Position'}}{\Var{\Next'}}\) to \(C
\triangleq \IContainer{\Var{\Shape}}{\Var{\Position}}{\Var{\Next}}\)
gives you an effective recipe to \emph{simulate} the behavior of
\(C'\) using \(C\): a choice of shape \(\aShape'\) in \(\Shape'\) is
translated to a choice of shape in \(\Shape\) through \(\MorphShape\)
while the subsequent response \(\TypeAnn{\aPosition}{\Position
  (\MorphShape\: \aShape')}\) is back-translated through
\(\MorphPosition\) to a response in \(\Position'\).


We shall need the following lemma that creates polynomial functors
from a cartesian natural transformations to a polynomial functor:
\begin{lemma}[Lemma 2.2~\citep{gambino:poly-monads}]
  \label{lemma:cart-induce-poly}

Let \(\TypeAnn{P_F}{\Slice{\LCCC{E}}{I} \to \Slice{\LCCC{E}}{J}}\) be
a polynomial functor. Let \(Q\) a functor from \(\Slice{\LCCC{E}}{I}\)
to \(\Slice{\LCCC{E}}{J}\).

If \(\TypeAnn{\phi}{Q \NatTrans P_F}\) is a cartesian natural
transformation, then \(Q\) is also a polynomial functor.

\end{lemma}


Finally, we shall need the following algebraic characterization of
polynomial functors:

\begin{lemma}[Corollary 1.14~\citep{gambino:poly-monads}]
  \label{lemma:polyfunc-smallest-class}
  The class of polynomial functors is the smallest class of functors
  between slices of \(\LCCC{E}\) containing the pullback functors and
  their adjoints, and closed under composition and natural
  isomorphism.
\end{lemma}


\subsection{Framed bicategory}


We have resisted the urge of defining a category of polynomials and
polynomial functors. Such a category can be defined for a given pair
of indices \(I\) and \(J\), with objects being polynomials indexed by
\(I\) and \(J\) (Definition~\ref{def:poly}) and morphisms
(Definition~\ref{def:poly-morphism}) specialised to the case
where \(\TypeAnn{\MorphSource = \CatId}{I \to I}\) and
\(\TypeAnn{\MorphTarget = \CatId}{J \to J}\).

From there, we are naturally lead to organise polynomials and their
indices in a 2-category. However, this fails to capture the fact that
indices have a life of their own: it makes sense to have morphisms
between differently indexed functors, \ie between different slices of
\(\LCCC{E}\). Indeed, morphisms between indices -- the objects --
\emph{induce} 1-morphisms. The 2-categorical presentation does not
capture this extra power.
Following the steps of \citet{gambino:poly-monads}, we organize
polynomials and their functors in a \emph{framed
  bicategory}~\citep{shulman:framed-bicat}. We refer the reader to
that latter paper for a comprehensive presentation of this concept.
%
%
A framed bicategory is a double category with some more structure. We
therefore recall the definition of a double category and give a few
examples.

\begin{definition}[Double category]

A double category \(\Cat{D}\) consists of a category of objects
\(\Cat{D}_0\) and a category of morphisms \(\Cat{D}_1\), together with
structure functors:
\[
\begin{tikzpicture}[description/.style={fill=white,inner sep=2pt}]
\matrix (m) [matrix of math nodes
            , row sep=3em
            , column sep=5em
            , text height=1.5ex
            , text depth=0.25ex
            , ampersand replacement=\&
            ]
{ 
    \Cat{D}_0  \& \Cat{D}_1 \& \Pullback{\Cat{D}_1}{\Cat{D}_0}{\Cat{D}_1} \\
};
\path[->] 
   (m-1-2) edge[bend right=30] node[above] { \(L\) } (m-1-1)
   (m-1-1) edge node[description] { \(U\) } (m-1-2) 
   (m-1-2) edge[bend left=30] node[below] { \(R\) } (m-1-1)
   (m-1-3) edge node[above] { \(\odot\) } (m-1-2)
;
\end{tikzpicture}
\]
Satisfying the usual axioms of categories relating the left frame
\(L\), right frame \(R\), identity \(U\), and composition
\(\odot\).

\end{definition}


\newcommand{\PolyFunc}[1]{{\mathit{PolyFun}}_{#1}}

\begin{example}[Double category \(\PolyFunc{\LCCC{E}}\)~{\citep[\S{3.5}]{gambino:poly-monads}}]

The double category \(\PolyFunc{\LCCC{E}}\) is defined as follow:
\begin{itemize}
\item Objects (\ie objects of \(\Cat{D}_0\)): slices \(\SliceL{E}{I}\)
\item Vertical arrows (\ie morphism of \(\Cat{D}_0\)): colift
  \(\TypeAnn{\CatExists{u}}{\SliceL{E}{I} \to \SliceL{E}{K}}\), for
  \(\TypeAnn{u}{I \to K}\) in \(\LCCC{E}\)

\item Horizontal arrows (\ie objects of \(\Cat{D}_1\)) with left frame
  \(\SliceL{E}{I}\) and right frame \(\SliceL{E}{J}\): polynomial
  functor \(\TypeAnn{P_F}{\SliceL{E}{I} \to \SliceL{E}{J}}\)

\item Squares (\ie morphisms of \(\Cat{D}_1\)) with left frame
  \(\CatExists{u}\) and right frame \(\CatExists{v}\): strong natural
  transformation \(\phi\):
\[
\begin{tikzpicture}
\matrix (m) [matrix of math nodes
            , row sep=3em
            , column sep=5em
            , text height=1.5ex
            , text depth=0.25ex
            , ampersand replacement=\&
            ]
{ 
    \SliceL{E}{I}    \&     \SliceL{E}{J} \\
    \SliceL{E}{K}    \&     \SliceL{E}{L} \\
};
\node at ($(m-1-1) !.5! (m-2-2)$) {\(\Downarrow \phi\)};
\path[->] 
   (m-1-1) edge node[above] { \(P_F\) } (m-1-2)
   (m-1-1) edge node[left] { \(\CatExists{u}\) } (m-2-1) 
   (m-1-2) edge node[right] { \(\CatExists{v}\) } (m-2-2)
   (m-2-1) edge node[below] { \(P_G\) } (m-2-2)
;
\end{tikzpicture}
\]
\end{itemize}

\end{example}


\newcommand{\Poly}[1]{{\mathit{Poly}}_{#1}}

\begin{example}[Double category \(\Poly{\LCCC{E}}\)]

The double category \(\Poly{\LCCC{E}}\) is defined as follow:
\begin{itemize}
\item \(\Cat{D}_0 = \Set\), \ie
  \begin{itemize}
  \item Objects: index set \(I, J, K, L, \ldots\)
  \item Vertical arrows: \(\TypeAnn{u}{I \to K}\), \(\TypeAnn{v}{J \to L}\), \ldots
  \end{itemize}
\item Horizontal arrows (\ie objects of \(\Cat{D}_1\)) with left frame \(I\) and
  right frame \(J\): polynomial indexed by \(I\) and \(J\)
\item Squares (\ie morphisms of \(\Cat{D}_1\)) with left frame
  \(\CatExists{u}\) and right frame \(\CatExists{v}\): polynomial
  morphism, with \(u\) closing the diagram on the left and \(v\)
  closing the diagram on the right.
\end{itemize}

\end{example}

We are naturally tempted to establish a connection between the double
category \(\Poly{\LCCC{E}}\) of polynomials and the double category
\(\PolyFunc{\LCCC{E}}\) of polynomial functors. We expect the
interpretation of polynomials to play that role, behaving, loosely
speaking, as a functor from \(\Poly{\LCCC{E}}\) to
\(\PolyFunc{\LCCC{E}}\). To formalize that intuition, we need a notion
of functor between double categories:
\begin{definition}[Lax double functor~{\citep[\S{6.1}]{shulman:framed-bicat}}]

A lax double functor \(\TypeAnn{F}{\Cat{C} \to \Cat{D}}\) consists of:
\begin{itemize}
  \item Two functors \(\TypeAnn{F_0}{\Cat{C}_0 \to \Cat{D}_0}\) and
    \(\TypeAnn{F_1}{\Cat{C}_1 \to \Cat{D}_1}\) such that \(L \circ F_1
    = F_0 \circ L\) and \(R \circ F_1 = F_0 \circ R\).
  \item Two natural transformations transporting the \(\odot\) and
    \(U\) functors in the expected way.
\end{itemize}

\end{definition}


Having presented the double-categorical framework, we now move on to
framed bicategories. The key intuition here comes from our earlier
observation: morphisms between indices, \ie morphisms in
\(\Cat{D}_0\), induce polynomials, \ie objects in
\(\Cat{D}_1\). Categorically, this translates into a bifibrational
structure on the \((L,R)\) functor:
\begin{definition}[Framed bicategory~\citep{shulman:framed-bicat}]

  A framed bicategory is a double category for which the functor
\[
\TypeAnn{(L, R)}{\Cat{D}_1 \to \Cat{D}_0 \times \Cat{D}_0}
\]
is a bifibration.

\end{definition}

\newcommand{\BaseChange}[3]{{(#2,#3)}^{*}#1}
\newcommand{\CobaseChange}[3]{(#2,#3)_{!}#1}

\begin{example}[Framed bicategory \(\Poly{\LCCC{E}}\)~{\citep[\S{3.7}]{gambino:poly-monads}}]

The bifibration is therefore the \emph{endpoints}
functor~\citep[\S{3.10}]{gambino:poly-monads} for which the cobase
change of polynomial \(\TypeAnn{F}{\PolyDiag{f}{B}{A}{s}{I}{t}{J}}\)
along \((u,v)\), denoted \(\CobaseChange{F}{u}{v}\), is
\[
\begin{tikzpicture}
\matrix (m) [matrix of math nodes
            , row sep=3em
            , column sep=5em
            , text height=1.5ex
            , text depth=0.25ex
            , ampersand replacement=\&
            ]
{ 
    I  \& B \& A \& J \\
    K  \& B \& A \& L  \\
};
\path[->] 
   (m-1-2) edge node[above] {\(\PolySource\)} (m-1-1)
   (m-1-2) edge node[above] {\(\Polynome\)} (m-1-3)
   (m-1-3) edge node[above] {\(\PolyTarget\)} (m-1-4)
   (m-2-2) edge node[below] { \(\MorphSource \circ \PolySource\) } (m-2-1)
   (m-2-2) edge node[below] { \(\Polynome\) } (m-2-3) 
   (m-2-3) edge node[below] { \(\MorphTarget \circ \PolyTarget\) } (m-2-4)
   (m-1-1) edge node[left] { \(\MorphSource\) } (m-2-1)
   (m-1-2) edge[double,-] (m-2-2)
   (m-1-3) edge[double,-] (m-2-3)
   (m-1-4) edge node[right] { \(\MorphTarget\) } (m-2-4)
;
\end{tikzpicture}
\]
While the base change of \(G\) along \((\MorphSource,\MorphTarget)\), denoted
\(\BaseChange{G}{u}{v}\) is defined as:
\[
\begin{tikzpicture}
\matrix (m) [matrix of math nodes
            , row sep=3em
            , column sep=5em
            , text height=1.5ex
            , text depth=0.25ex
            , ampersand replacement=\&
            ]
{ 
    I     \&       \& \cdot \& \cdot \& J    \\
    \cdot \&       \& \cdot \& \cdot \& \cdot \\
    \cdot \& \cdot \&       \&       \&       \\
    K     \&       \& D     \& C     \& L     \\
};
\path[->] 
   (m-1-3) edge node[above] {} (m-1-1)
   (m-1-3) edge node[above] {} (m-1-4) 
   (m-1-4) edge node[above] {} (m-1-5)
   (m-2-3) edge node[above] {} (m-2-1)
   (m-2-3) edge node[above] {} (m-2-4)
   (m-2-4) edge node[above] {} (m-2-5)
   (m-4-3) edge node[below] {\(\PolySource\)} (m-4-1)
   (m-4-3) edge node[below] {\(\Polynome\)} (m-4-4)
   (m-4-4) edge node[below] {\(\PolyTarget\)} (m-4-5)
   (m-1-1) edge[double,-] (m-2-1)
   (m-1-3) edge node[right] {\(\MorphTarget^{\dagger\dagger}\)} (m-2-3)
   (m-1-4) edge node[right] {\(\MorphTarget^\dagger\)} (m-2-4)
   (m-1-5) edge node[right] {\(\MorphTarget\)} (m-2-5)
   (m-2-1) edge[double,-] (m-3-1)
   (m-3-1) edge node[left] {\(\MorphSource\)} (m-4-1)
   (m-2-3) edge node[below] {\(\epsilon\)} (m-3-2)
   (m-3-2) edge node[above] {\(\MorphSource^\dagger\)} (m-4-3)
   (m-3-2) edge node[above] {\(\PolySource^\dagger\)} (m-3-1)
   (m-2-3) edge node[right] {\(\Reindex{\Polynome} \CatForall{\Polynome} \MorphSource^\dagger\)} (m-4-3)
   (m-2-4) edge node[right] {\(\CatForall{\Polynome} \MorphSource^\dagger\)} (m-4-4)
   (m-2-5) edge[double,-] (m-4-5)
;
\begin{scope}[shift=($(m-1-3)!.25!(m-2-4)$)]
\draw +(-.25,0) -- +(0,0)  -- +(0,.25);
\end{scope}
\begin{scope}[shift=($(m-1-4)!.25!(m-2-5)$)]
\draw +(-.25,0) -- +(0,0)  -- +(0,.25);
\end{scope}
\begin{scope}[shift=($(m-2-3)!.25!(m-3-4)$)]
\draw +(-.25,0) -- +(0,0)  -- +(0,.25);
\end{scope}
\begin{scope}[shift=($(m-3-2)!.25!(m-4-1)$)]
\draw +(0,.25) -- +(0,0) -- + (.25,0);
\end{scope}
\end{tikzpicture}
\]

\end{example}
As before, these definitions straightforwardly translates to operations
on containers.



\begin{example}[Framed bicategory \(\PolyFunc{\LCCC{E}}\)~{\citep[\S{3.6}]{gambino:poly-monads}}]

  The fibrational structure of the framed bicategory gives rise to a
  \emph{transporter lift} (the cartesian lifting of the fibration
  \((L, R)\)) and a \emph{cotransporter lift} (the op-cartesian
  lifting of the op-fibration \((L,R)\)).

  The transporter lift of \((u,v)\) to \(P_G\) is given by:
\[
\begin{tikzpicture}
\matrix (m) [matrix of math nodes
            , row sep=3em
            , column sep=5em
            , text height=1.5ex
            , text depth=0.25ex
            , ampersand replacement=\&
            ]
{ 
    \SliceL{E}{I}  \& \SliceL{E}{K} \& \SliceL{E}{L} \& \SliceL{E}{K} \\
    \SliceL{E}{K}  \& \SliceL{E}{K} \& \SliceL{E}{L} \& \SliceL{E}{L} \\
};
\node at ($(m-1-1) !.5! (m-2-2)$) {\(\Iso\)};
\node at ($(m-1-2) !.5! (m-2-3)$) {\(\Iso\)};
\node at ($(m-1-3) !.5! (m-2-4)$) {\(\Downarrow \eta\)};
\path[->] 
   (m-1-1) edge node[above] { \(\CatExists{u}\) } (m-1-2)
   (m-1-2) edge node[above] { \(P_G\) } (m-1-3) 
   (m-1-3) edge node[above] { \(\Reindex{v} \) } (m-1-4)
   (m-2-1) edge[double,-] (m-2-2)
   (m-2-2) edge node[below] {\(P_G\)} (m-2-3)
   (m-2-3) edge[double,-] (m-2-4)
   (m-1-1) edge node[left] { \(\CatExists{u}\) } (m-2-1)
   (m-1-2) edge[double,-] (m-2-2)
   (m-1-3) edge[double,-] (m-2-3)
   (m-1-4) edge node[right] { \(\CatExists{v}\) } (m-2-4)
;
\end{tikzpicture}
\]

The cotransporter lift of \((u,v)\) to \(P_F\) is given by:
\[
\begin{tikzpicture}
\matrix (m) [matrix of math nodes
            , row sep=3em
            , column sep=5em
            , text height=1.5ex
            , text depth=0.25ex
            , ampersand replacement=\&
            ]
{ 
    \SliceL{E}{I} \& \SliceL{E}{I} \& \SliceL{E}{J} \& \SliceL{E}{J} \\
    \SliceL{E}{K} \& \SliceL{E}{I} \& \SliceL{E}{J} \& \SliceL{E}{L} \\
};
\node at ($(m-1-1) !.5! (m-2-2)$) {\(\Downarrow \epsilon\)};
\node at ($(m-1-2) !.5! (m-2-3)$) {\(\Iso\)};
\node at ($(m-1-3) !.5! (m-2-4)$) {\(\Iso\)};
\path[->] 
   (m-1-1) edge[double,-] (m-1-2)
   (m-1-2) edge node[above] {\(P_F\)} (m-1-3)
   (m-1-3) edge[double,-] (m-1-4)
   (m-2-1) edge node[below] { \(\Reindex{u}\) } (m-2-2)
   (m-2-2) edge node[below] { \(P_F\) } (m-2-3) 
   (m-2-3) edge node[below] { \(\CatExists{v} \) } (m-2-4)
   (m-1-1) edge node[left] { \(\CatExists{u}\) } (m-2-1)
   (m-1-2) edge[double,-] (m-2-2)
   (m-1-3) edge[double,-] (m-2-3)
   (m-1-4) edge node[right] { \(\CatExists{v}\) } (m-2-4)
;
\end{tikzpicture}
\]

  Following \citet{gambino:poly-monads}, we define the base change of
  \(P_G\) along \((u,v)\) by \(\BaseChange{P_G}{u}{v} = \Reindex{v} \circ
  P_G \circ \CatExists{u}\). Dually, we define the cobase change of
  \(P_F\) along \((u,v)\) by \(\CobaseChange{P_F}{u}{v} = \CatExists{v}
  \circ P_F \circ \Reindex{u}\).

\end{example}


At this stage, it should be clear that the interpretation of
polynomials is more than a mere functor from \(\Poly{\LCCC{E}}\) to
\(\PolyFunc{\LCCC{E}}\): loosely speaking, it establishes an
equivalence of categories. Equivalence of framed bicategory is
formally defined as follow:
\begin{definition}[Framed biequivalence~{\citep[\S{7.1}]{shulman:framed-bicat}}]

  A framed equivalence between the framed bicategory \(\Cat{C}\) and
  \(\Cat{D}\) consists of:
  \begin{itemize}
  \item Two lax double functors \(\TypeAnn{F}{\Cat{C} \to
    \Cat{D}}\) and \(\TypeAnn{G}{\Cat{D} \to \Cat{C}}\), and
  \item Two framed natural isomorphism \(\TypeAnn{\eta}{1 \Iso G \circ
    F}\) and \(\TypeAnn{\epsilon}{F \circ G \Iso 1}\)
  \end{itemize}

\end{definition}


We then recall this result of~\citet{gambino:poly-monads}:
\begin{theorem}[Theorem 3.8~\citep{gambino:poly-monads}]

  Squares of \(\PolyFunc{\LCCC{E}}\) are uniquely represented (up to a choice of
  pullback) by a square of \(\Poly{\LCCC{E}}\). Consequently, the interpretation of
  polynomials is a framed biequivalence.

\end{theorem}

The interpretation functor is an equivalence of framed bicategory
between \(\Poly{\LCCC{E}}\) and the framed bicategory
\(\PolyFunc{\LCCC{E}}\). We
thus conflate the category of polynomials \(\Poly{\LCCC{E}}\) and the
category of polynomial functors \(\PolyFunc{\LCCC{E}}\). Polynomials
are a ``small'' presentation of the larger functorial objects. Since
both categories are equivalent, we do not lose expressive power by
working in the small language.


\newcommand{\PolyC}[1]{{\Poly{\!}}^c_{#1}}
\newcommand{\PolyFuncC}[1]{{\PolyFunc{\!}}^c_{#1}}

Earlier, we have isolated a class of cartesian morphisms. This defines
a sub-category of polynomials, which will be of prime interest in this
paper. For clarity, we expound its definition:
\begin{example}[Framed bicategory \(\PolyC{\LCCC{E}}\)~{\citep[\S{3.13}]{gambino:poly-monads}}]
  \label{example:polyc-bicat}

  The framed bicategory \(\PolyC{\LCCC{E}}\) is defined by:
  \begin{itemize}
  \item Objects: indices, \ie objects of \(\LCCC{E}\)
  \item Vertical arrows: index morphisms, \ie morphisms of
    \(\LCCC{E}\)
  \item Horizontal arrows: polynomial indexed by \(I\) and \(J\), respectively left and right frames
  \item Squares: cartesian morphism of polynomial reindexed by
    \(u\) and \(v\), respectively left and right frames:
\[
\begin{tikzpicture}
\matrix (m) [matrix of math nodes
            , row sep=3em
            , column sep=5em
            , text height=1.5ex
            , text depth=0.25ex
            , ampersand replacement=\&]
{ 
    I  \& B  \& A  \& J \\
    K  \& D  \& C  \& L \\
};
\path[->] 
   (m-1-2) edge (m-1-1)
   (m-1-2) edge (m-1-3)
   (m-1-3) edge (m-1-4)
   (m-2-2) edge (m-2-1)
   (m-2-2) edge (m-2-3)
   (m-2-3) edge (m-2-4)
   (m-1-1) edge node[left] { \(\PolySource\) } (m-2-1)
   (m-1-4) edge node[right] { \(\PolyTarget\) } (m-2-4)
   (m-1-2) edge (m-2-2)
   (m-1-2) edge (m-1-3)
   (m-1-3) edge node[right] { \(\MorphPull\) } (m-2-3)
;
\begin{scope}[shift=($(m-1-2)!.25!(m-2-3)$)]
\draw +(-.25,0) -- +(0,0)  -- +(0,.25);
\end{scope}
\end{tikzpicture}
\]
  \end{itemize}
  That is, we take \(\Cat{D}_0 \triangleq \LCCC{E}\) and \(\Cat{D}_1
  \triangleq \biguplus_{I,J} \PolyC{\LCCC{E}}(I,J)\) for which we define:
  \begin{itemize}
  \item The identity functor \(U\) that maps an index to the identity
    polynomial at that index~;
  \item The left frame \(L\) that projects the source index \(I\) ;
  \item The right frame \(R\) that projects the target index \(J\) ;
  \item The composition \(\odot\) of polynomial functors.
  \end{itemize}
  The \emph{frames} defined by \(L\) and \(R\) thus correspond to,
  respectively, the left-hand side and right-hand side of polynomials
  and polynomial morphisms. As for the bifibration structure, consider
  a pair of morphism \(\TypeAnn{u, v}{K \to I, L \to J}\) in the base
  category \(\LCCC{E} \times \LCCC{E}\), we have:
  \begin{itemize}
  \item A cobase-change functor reindexing a polynomial
    \(\TypeAnn{P}{\PolyC{\LCCC{E}}(I,J)}\) to a polynomial
    \(\TypeAnn{\CobaseChange{P}{u}{v}}{\PolyC{\LCCC{E}}(K,L)}\) ;
  \item A base-change functor reindexing a polynomial
    \(\TypeAnn{P}{\PolyC{\LCCC{E}}(K,L)}\) to a polynomial
    \(\TypeAnn{\BaseChange{P}{u}{v}}{\PolyC{\LCCC{E}}(I,J)}\).
  \end{itemize}
  This extra-structure lets us transport polynomials across frames:
  given a polynomial, we can reindex or op-reindex it to any frame
  along a pair of index morphisms.

\end{example}

The interpretation functor is again an equivalence of framed
bicategory between \(\PolyC{\LCCC{E}}\) and the framed bicategory
\(\PolyFuncC{\LCCC{E}}\) which morphisms of functors consists of
cartesian natural transformations. We therefore have the following
result:
\begin{theorem}[Theorem 3.13~\citep{gambino:poly-monads}]

  The subcategory \(\PolyFuncC{\LCCC{E}}\) is framed biequivalent to
  \(\PolyC{\LCCC{E}}\). In particular, squares of
  \(\PolyFuncC{\LCCC{E}}\) are uniquely represented by their
  diagrammatic counterpart in \(\PolyC{\LCCC{E}}\).

\end{theorem}

The interpretation functor is an equivalence of framed bicategory
between \(\PolyC{\LCCC{E}}\) and the framed bicategory
\(\PolyFuncC{\LCCC{E}}\)~\citep[Theorem 3.13]{gambino:poly-monads}. We
can therefore conflate, once and for all, the category of polynomials
\(\PolyC{\LCCC{E}}\) and the category of polynomial functors
\(\PolyFuncC{\LCCC{E}}\). In a sense, polynomials are a ``small''
presentation of the larger functorial objects. However, we do not lose
expressive power by working in the small language, since both
categories are equivalent.


\section{Inductive Families in Type Theory}
\label{sec:desc-poly}


In this section, we set out to establish a formal connection between a
presentation of inductive families in type theory and the categorical
model of polynomial functors. On the type theoretical side, we adopt
the universe-based presentation introduced by
\citet{dagand:levitation}. Working on a universe gives us a syntactic
internalisation of inductive families within type theory. Hence, we
can manipulate and reason about inductive families from within the
type theory itself.

The original motivation for this design is \emph{generic programming}:
the programmer can compute over the structure of datatypes, or even
compute new datatypes from old. In mathematical term, ``generic
programming'' reads as \emph{reflection}: we can reflect the meta-theory
of inductive types within the type theory. For systems like Agda or
Coq, we can imagine reducing their syntactic definition to such a
universe by \emph{elaboration}~\citep{dagand:elaboration}.


We recall the definition of the universe in
\figurename~\ref{fig:universe-desc}. A \(\IDesc\) code is a syntactic
object \emph{desc}ribing a functor from \(\Set^{\Var{I}}\) to
\(\Set\). To obtain this functor, we have to interpret the code using
\(\InterpretIDesc{\_}\). The reader will gain intuition for the codes
by looking at their interpretation, \ie their semantics. To describe
functors from \(\Set^{\Var{I}}\) to \(\Set^{\Var{J}}\), we use the
isomorphism \([\Set^{\Var{I}}, \Set]^{\Var{J}} \Iso [\Set^{\Var{I}},
  \Set^{\Var{J}}]\). Hence, in \(\idesc\), we pull the \(J\)-index to
the front and thus capture functors on slices of \(\Set\). The
interpretation \(\InterpretIDesc{\_}\) extends pointwise to
\(\idesc\). Inhabitants of the \(\idesc\) type are called
\emph{descriptions}. By construction, the interpretation of a
description is a strictly positive functor: for a description \(D\),
the initial algebra always exists and is denoted \((\lfp{D},
\TypeAnn{\inLfp}{\InterpretIDesc{D}[\mu{D}] \to \mu{D}})\).

\begin{figure}[!t]

{\footnotesize
\[
\begin{array}{l@{\qquad\qquad\qquad}l}
\IDescDef & \InterpretIDescDef\\
\medskip & \\
\Let{\idesc &
         \PiTel{\Var{I}}{\Set} &
         \PiTel{\Var{J}}{\Set}}{\Set[1]}{
\Return{\idesc & \Var{I} & \Var{J}}{\Var{J} \To \IDesc[\Var{I}]}}
&
\Let{\InterpretDesc{\PiTel{\Var{D}}{\idesc[\Var{I}\: \Var{J}]}} &
     \PiTel{\Var{X}}{\Var{I} \To \Set}}
    {\Var{J} \To \Set}{
\Return{\InterpretDesc{\Var{D}} & \Var{X}}{\Lam{\Var{j}}{\InterpretDesc{\Var{D}\: \Var{j}}\: \Var{X}}}
}
\end{array}
\]}

\caption{Universe of inductive families}
\label{fig:universe-desc}

\end{figure}

\begin{definition}[Described functor]

  A functor is \emph{described} if it is isomorphic to the
  interpretation of a description.

\end{definition}

\begin{example}[Natural numbers]

The signature functor of natural numbers is described by:
\[
\Let{\NatD}{\idesc[\Unit\: \Unit]}{
\Return{\NatD}
       {\Lam{\Void}
         \DSigma[{(\Unit \Sum \Unit)\:
                \LambdaBinder
                      \left\{
                      \begin{array}{@{}l@{\DoReturn}l}
                        \InjLeft[\Void]   & \DUnit \\
                        \InjRight[\Void]  & \DVar[\Void]
                      \end{array}
                      \right.
                    }]}}
\]
The reader will check that the interpretation \(\InterpretIDesc{\_}\)
of this code gives a functor isomorphic to the expected \(X \mapsto
1 + X\).

\end{example}


\subsection{Descriptions are equivalent to polynomials}

\newcommand{\ReindexDesc}[1]{\green{D\Reindex{\black{#1}}}}
\newcommand{\CatExistsDesc}[1]{\green{D\CatExists{\black{#1}}}}
\newcommand{\CatForallDesc}[1]{\green{D\CatForall{\black{#1}}}}
\newcommand{\ComposeDesc}{\mathop{\green{\circ_D}}}
\Spacedcommand{\ComposeDescHelp}{\Function{compose}}


We can now prove the equivalence between described functors and
polynomial functors. The first step is to prove that described
functors are polynomial:
\begin{lemma}
  \label{lemma:desc-sub-polyfun}

The class of described functors is included in the class of polynomial
functors.

\end{lemma}

\begin{proof}

Let \(\TypeAnn{F}{\Slice{\Set}{I} \to \Slice{\Set}{J}}\) a described
functor.

By definition of the class of described functor, \(F\) is naturally
isomorphic to the interpretation of a description. That is, for any
\(\TypeAnn{j}{J}\), there is a \(\TypeAnn{D}{\IDesc[I]}\) such that:
\[
F\: j \Iso \InterpretIDesc{D} 
\]


By induction over \(D\), we show that \(\InterpretIDesc{D}\) is
naturally isomorphic to a polynomial:
\begin{description}
\item[Case \(D = \DUnit\):] 

We have \(\InterpretIDesc{\DUnit}[X] \Iso \Unit \Times X^{\Empty}\),
which is clearly polynomial

\item[Case \(D = \DVar\: i\):] 

We have \(\InterpretIDesc{\DVar[i]}[X] \Iso \Unit \times (X\:
i)^{\Unit}\), which is clearly polynomial

\item[Case \(D = \DSigma\: S\: T\):] 

We have \(\InterpretIDesc{\DSigma[S\: T]}[X] =
\SigmaTimes{\Var{s}}{S}{\InterpretIDesc{T\: \Var{s}}[X]}\). By
induction hypothesis, \(\InterpretIDesc{T\: \Var{s}}[X] \Iso
\SigmaTimes{\Var{x}}{S_{T\: \Var{s}}}{\PiTo{\Var{p}}{P_{T\:
      \Var{s}}\: \Var{x}}{X\: (n_{T\: \Var{s}}\: p)}}\). Therefore, we obtain that:
\begin{align*}
\InterpretIDesc{\DSigma[S\: T]}[X] & \Iso 
    \SigmaTimes{\Var{s}}
               {S}
               {\SigmaTimes{\Var{x}}
                           {S_{T\: \Var{s}}}
                           {\PiTo{\Var{p}}
                                 {P_{T\: \Var{s}}}
                                 {X\: (n_{T\: \Var{s}}\: p)}}} \\
    & \Iso 
    \SigmaTimes{\Var{sx}}
               {\SigmaTimes{\Var{s}}{S}{S_{T\: \Var{s}}}}
               {\PiTo{\Var{p}}
                     {P_{T\: (\Fst[\Var{sx}])} (\Snd[\Var{sx}])}
                     {X\: (n_{T\: (\Fst[\Var{sx}])}\: \Var{p})}}
\end{align*}
This last functor being clearly polynomial.

\item[Case \(D = \DPi\: S\: T\):]

We have \(\InterpretIDesc{\DPi[S\: T]}[X] =
\PiTo{\Var{s}}{S}{\InterpretIDesc{T\: \Var{s}}[X]}\). By induction
hypothesis, \(\InterpretIDesc{T\: \Var{s}}[X] \Iso
\SigmaTimes{\Var{x}}{S_{T\: \Var{s}}}{\PiTo{\Var{p}}{P_{T\: \Var{s}}\:
    \Var{x}}{X\: (n_{T\: \Var{s}}\: p)}}\). Therefore, we obtain that:
\begin{align*}
\InterpretIDesc{\DPi[S\: T]}[X] & \Iso
    \PiTo{\Var{s}}{S} 
        \SigmaTimes{\Var{x}}
                   {S_{T\: \Var{s}}}
                   {\PiTo{\Var{p}}
                         {P_{T\: \Var{s}}\: \Var{x}}
                         {X\: (n_{T\: \Var{s}}\: p)}} \\
    & \Iso
    \SigmaTimes{\Var{f}}
               {\PiTo{\Var{s}}{S}{S_{T\: \Var{s}}}}
               {\PiTel{\Var{s}}{S}
                \PiTo{\Var{p}}{P_{T\: \Var{s}}\: (\Var{f}\: \Var{s})}
                 X\: (n_{T\: \Var{s}}\: \Var{p})} \\
    & \Iso
    \SigmaTimes{\Var{f}}
               {\PiTo{\Var{s}}{S}{S_{T\: \Var{s}}}}
               {\PiTo{\Var{sp}}
                     {\SigmaTimes{\Var{s}}{S}{P_{T\: \Var{s}}\: (\Var{f}\: \Var{s})}}
                     {X\: (n_{T\: (\Fst[\Var{sp}])}\: (\Snd[\Var{sp}]))}}
\end{align*}
This last functor being clearly polynomial.

\end{description}

\end{proof}



To prove the other inclusion -- that polynomials functors on \(\Set\) are a
subset of described functors -- we rely on
Lemma~\ref{lemma:polyfunc-smallest-class}. To this end, we must prove
some algebraic properties of the class of described functors, namely
that they are closed under reindexing, its adjoints, and
composition. To do so, the methodology is simply to code these
operations in \(\idesc\).

\begin{lemma}
  \label{lemma:desc-closed-reindexing}

Described functors are closed under reindexing and its adjoints.

\end{lemma}

\begin{proof}

We describe the pullback functors and their adjoints by:
\[
\Code[c]{
\Let{\ReindexDesc{\PiTel{\Var{f}}{\Var{A} \To \Var{B}}}}
    {\idesc[\Var{B}\: \Var{A}]}{
\Return{\ReindexDesc{\Var{f}}}
       {\Lam{\Var{a}}{\DVar[(\Var{f}\: \Var{a})]}}
} \\
\Let{\CatExistsDesc{\PiTel{\Var{f}}{\Var{A} \To \Var{B}}}}
    {\idesc[\Var{A}\: \Var{B}]}{
\Return{\CatExistsDesc{\Var{f}}}
       {\Lam{\Var{b}}
         \DSigma[\Var{f} \Inverse \Var{b}\: \Lam{\Var{a}}
                 \DVar[\Var{a}]]}
} \qquad
\Let{\CatForallDesc{\PiTel{\Var{f}}{\Var{A} \To \Var{B}}}}
    {\idesc[\Var{A}\: \Var{B}]}{
\Return{\CatForallDesc{\Var{f}}}
       {\Lam{\Var{b}}
         \DPi[\Var{f} \Inverse \Var{b}\: \Lam{\Var{a}}
             \DVar[\Var{a}]]}
}
}
\]
Where the inverse of a function \(f\) is represented by the following
inductive type:
\[
\InverseDef
\]

It is straightforward to check that these descriptions interpret to
the expected operation on slices of \(\Set\), \ie that we have:
\[
\InterpretIDesc{\ReindexDesc{f}} \Iso \Reindex{f} \qquad
\InterpretIDesc{\CatExistsDesc{f}} \Iso \CatExists{f} \qquad
\InterpretIDesc{\CatForallDesc{f}} \Iso \CatForall{f}
\]

\end{proof}


\begin{lemma} 
  \label{lemma:desc-closed-composition}

Described functors are closed under composition.
\end{lemma}

\begin{proof}

We define composition of descriptions by:
\[
\Code{
\Let{\PiTel{\Var{D}}{\idesc[\Var{B}\: \Var{C}]} &
     \ComposeDesc &
     \PiTel{\Var{E}}{\idesc[\Var{A}\: \Var{B}]}}
    {\idesc[\Var{A}\: \Var{C}]}{
\Return{\Var{D} & 
        \ComposeDesc &
        \Var{E}}
       {\Lam{\Var{c}}\ComposeDescHelp[(\Var{D}\: \Var{c})\: \Var{E}]}
} \\
\quad \where \\
\begin{array}{@{\qquad}l}
\Let{\ComposeDescHelp & 
     \PiTel{\Var{D}}{\IDesc[\Var{B}]} &
     \PiTel{\Var{E}}{\idesc[\Var{B}\: \Var{A}]}}
    {\IDesc[\Var{A}]}{
\Return{\ComposeDescHelp & 
          (\DVar[\Var{b}]) & 
          \Var{E}}
       {\Var{E}\: \Var{b}}
\Return{\ComposeDescHelp &
          \DUnit & 
          \Var{E}}
       {\DUnit}
\Return{\ComposeDescHelp &
          (\DPi[\Var{S}\: \Var{T}]) & 
          \Var{E}}
       {\DPi[\Var{S}\: \Lam{\Var{s}}{\ComposeDescHelp[(\Var{T}\: \Var{s})\: \Var{E}]}]}
\Return{\ComposeDescHelp &
          (\DSigma[\Var{S}\: \Var{T}]) &
          \Var{E}}
       {\DSigma[\Var{S}\: \Lam{\Var{s}}{\ComposeDescHelp[(\Var{T}\: \Var{s})\: \Var{E}]}]}
}
\end{array}
}
\]

It is then straightforward to check that this is indeed computing the
composition of the functors, \ie that we have:
\[
\InterpretIDesc{D \ComposeDesc E} \Iso \InterpretIDesc{D} \circ \InterpretIDesc{E}
\]

\end{proof}


\begin{lemma}
  \label{lemma:polyfun-sub-desc}

The class of polynomial functors is included in the class of described
functors.

\end{lemma}

\begin{proof}

Described functors are closed under reindexing, its left and right
adjoint (Lemma~\ref{lemma:desc-closed-reindexing}), are closed under
composition (Lemma~\ref{lemma:desc-closed-composition}) and are
defined up to natural isomorphism. By
Lemma~\ref{lemma:polyfunc-smallest-class}, the class of
polynomial functors is the least such set. Therefore, the class of
polynomial functor is included in the class of described functors.

\end{proof}


We conclude with the desired equivalence:
\begin{proposition}
  \label{prop:desc-iso-polyfun}
  The class of described functors corresponds exactly to the class of
  polynomial functors. 
\end{proposition}

\begin{proof}

By Lemma~(\ref{lemma:desc-sub-polyfun}) and
Lemma~(\ref{lemma:polyfun-sub-desc}), we have both inclusions.

\end{proof}


The benefit of this algebraic approach is its flexibility with respect
to the universe definition: for practical purposes, we are likely to
introduce new \(\IDesc\) codes. However, the implementation of
reindexing and its adjoints will remain unchanged. Only composition
would need to be verified. Besides, these operations are useful in
practice, so we are bound to implement them anyway. In the rest of
this paper, we shall conflate descriptions, polynomials, and
polynomial functors, silently switching from one to another as we see
fit.


\subsection{An alternative proof}

\newcommand{\IDescToIS}[1]{\green{\langle}#1\green{\rangle}}
\Spacedcommand{\IDescToShape}{\Function{Shape}}
\Spacedcommand{\IDescToPos}{\Function{Pos}}
\Spacedcommand{\IDescToIndex}{\Function{Index}}

\newcommand{\ISToIDesc}[1]{\green{\langle}#1\green{\rangle^{-1}}}


An alternative approach, followed by~\citet{morris:PhD} for
example, consists in reducing these codes to containers. We thus
obtain the equivalence to polynomial functors, relying on the fact
that containers are an incarnation of polynomial functors in the
internal language~\citet[\S{2.18}]{gambino:poly-monads}. This less
algebraic approach is more constructive. However, to be absolutely
formal, it calls for proving some rather painful (extensional)
equalities.
%
%
%
If the proofs are laborious, the translation itself is not devoid of
interest. In particular, it gives an intuition of descriptions in
terms of shape, position and indices. This slightly more abstract
understanding of our universe will be useful in this paper, and is
useful in general when reasoning about datatypes.

We formalise the translation in Figure~\ref{fig:desc-to-cont}, mapping
descriptions to containers. The message to take away from that
translation is which code contributes to which part of the container,
\ie shape, position, and/or index. Crucially, the \(\DUnit\) and
\(\DSigma\) codes contribute only to the shapes. The \(\DVar\) and
\(\DPi\) codes, on the other hand, contribute to the
positions. Finally, the \(\DVar\) code is singly defining the next
index.

The inverse translation is otherwise trivial and given here for the
sake of completeness:
\[
\Let{\ISToIDesc{\PiTel{\Var{\aContainer}}{\ICont[\Var{I}\: \Var{J}]}}}
    {\idesc[\Var{I}\: \Var{J}]}{
\Return{\ISToIDesc{\IContainer{\Var{\Shape}}{\Var{\Position}}{\Var{\Next}}}}
       {\DSigma[{\Var{\Shape}\: \Lam{\Var{\aShape}}
        \DPi[{(\Var{\Position}\: \Var{\aShape})\: \Lam{\Var{\aPosition}}
        \DVar[(\Var{\Next}\: \Var{\aPosition})]}]}]}
}
\]
We are left to prove that these translations are
indeed inverse of each other: while this proof is extremely tedious to
carry formally, it should be intuitively straightforward. We
therefore assume the following lemma:
\begin{lemma}[Described functors to polynomials, alternatively]
  \label{lemma:idesc-to-is}

\(\IDescToIS{\_}\) is essentially surjective.

\end{lemma}

\begin{figure}

{\footnotesize
\[\Code{
\Code{
\Let{\IDescToIS{\PiTel{\Var{D}}{\idesc[\Var{I}\: \Var{J}]}}}
    {\ICont[\Var{I}\: \Var{J}]}{}\\
\Case{
\Return{\IDescToIS{\Var{D}}}
       {\IContainer{\Lam{\Var{j}}\IDescToShape[(\Var{D}\: \Var{j})]}
                   {\Lam{\Var{j}}\IDescToPos[(\Var{D}\: \Var{j})]}
                   {\Lam{\Var{j}}\IDescToIndex[(\Var{D}\: \Var{j})]}}}
\quad \where}\\
\qquad\Code{
  \Let{\IDescToShape & \PiTel{\Var{D}}{\IDesc[\Var{I}]}}
      {\Set}{
  \Return{\IDescToShape & \DVar[\Var{i}]}
         {\Unit}
  \Return{\IDescToShape & \DUnit}
         {\Unit}
  \Return{\IDescToShape & \DPi[\Var{S}\: \Var{T}]}
         {\PiTo{\Var{s}}{\Var{S}}{\IDescToShape[(\Var{T}\: \Var{s})]}}
  \Return{\IDescToShape & \DSigma[\Var{S}\: \Var{T}]}
         {\SigmaTimes{\Var{s}}{\Var{S}}{\IDescToShape[(\Var{T}\: \Var{s})]}}
  } \\ \\
  \Let{\IDescToPos & \PiTel{\Var{D}}{\IDesc[\Var{I}]}
                   & \PiTel{\Var{\aShape}}{\IDescToShape[\Var{D}]}}
      {\Set}{
  \Return{\IDescToPos & \DVar[\Var{i}] & \Void}
         {\Unit}
  \Return{\IDescToPos & \DUnit & \Void}
         {\Empty}
  \Return{\IDescToPos & \DPi[\Var{S}\: \Var{T}] & \Var{f}}
         {\SigmaTimes{\Var{s}}{\Var{S}}{\IDescToPos[(\Var{T}\: \Var{s})\: (\Var{f}\: \Var{s})]}}
  \Return{\IDescToPos & \DSigma[\Var{S}\: \Var{T}] & \Pair{\Var{s}}{\Var{t}}}
         {\IDescToPos[(\Var{T}\: \Var{s})\: \Var{t}]}
  } \\ \\
  \Let{\IDescToIndex & \PiTel{\Var{D}}{\IDesc[\Var{I}]}
                     & \PiTel{\Var{\aPosition}}{\IDescToPos[\Var{D}\: \Var{\aShape}]}}
      {\Var{I}}{
  \Return{\IDescToIndex & \DVar[\Var{i}] & \Void}
         {\Var{i}}
  \Return{\IDescToIndex & \DPi[\Var{S}\: \Var{T}] & \Pair{\Var{s}}{\aPosition}}
         {\IDescToIndex[(\Var{T}\: \Var{s})\: \Var{\aPosition}]}
  \Return{\IDescToIndex & \DSigma[\Var{S}\: \Var{T}] & \Var{\aPosition}}
         {\IDescToIndex[(\Var{T}\: (\Fst[\Var{sh}]))\: \Var{\aPosition}]}
  } 
} 
}\]
}

\caption{From descriptions to containers}
\label{fig:desc-to-cont}

\end{figure}


\subsection{Discussion}


Let us reflect on the results obtained in this section. By
establishing an equivalence between descriptions -- a programming
artefact -- and polynomial functors -- a mathematical object -- we
connect software to mathematics, and conversely. On the one hand,
descriptions are suitable for practical purposes: they are a syntactic
object, fairly intensional, and can therefore be conveniently
manipulated by a computer. Polynomial functors, on the other hand, are
fit for theoretical work: they admit a diagrammatic representation and
are defined extensionally, up to natural isomorphism.

Better still, we have introduced containers as a middle ground between
these two presentations. Containers are an incarnation of polynomials
in the internal language. Reasoning extensionally about them is
equivalent to reasoning about polynomials. Nonetheless, they are also
rather effective type theoretic procedures: we can implement them in
Agda\footnote{Such an implementation is available on the the first
  author's website.}.


We shall traverse this bridge between software and mathematics in both
directions. Going from software to mathematics, we hope to gain a
deeper understanding of our constructions. Case in point is generic
programming in type theory: we develop many constructions over
datatypes, such as ornaments, but the justification for these is often
extremely operational, one might even say ``ad-hoc''. By putting our polynomial
glasses on, we can finally see through the syntax and understand the
structure behind these definitions. Conversely, going from mathematics
to software, we translate mathematical structures to new software
constructions. The theory of polynomial functors is indeed well
developed. Most programming examples presented in this paper -- such
as derivatives or ornaments -- were first presented in the polynomial
functor literature. Besides, by exploring the structure of polynomial
functors, we discover new and interesting programming idioms -- such
as the pullback and composition of ornaments.


The categorically minded reader might be tempted to look for an
equivalence of category. However, we have not yet introduced any
notion of morphism between descriptions. What we have established is a
lowly ``set theoretic'' equivalence between the class of descriptions
and the class of polynomial functors. In terms of equivalence of
categories, we have established that the object part of a functor, yet
to be determined, maps descriptions to polynomial functors in an
essentially surjective way. We shall complete this construction in the
following section. We will set up descriptions in a double category
with ornaments as morphisms. The translation \(\IDescToIS{\_}\) will
then functorially map it to the double category
\(\PolyFuncC{\LCCC{E}}\).


\section{A Categorical Treatment of Ornaments}
\label{sec:orn}


\newcommand{\ListO}{\Function{List-Orn}}
\newcommand{\FinO}{\Function{Fin-Orn}}

The motivation for ornaments comes from the frequent need, when using
dependent types, to relate datatypes that share the same structure. In
this setting, ornaments play the role of an organisation
principle. Intuitively, an ornament is the combination of two datatype
transformations: we may \emph{extend} the constructors, and/or
\emph{refine} the indices. Ornaments preserve the underlying
data-structure by enforcing that an extension respects the arity of
the original constructors. By extending a datatype, we introduce more
information, thus enriching its logical content. A typical example of
such an ornament is the one taking natural numbers to lists:
\[
\NatDef 
\quad
\stackrel{\ListO\:\Var{A}}{\Rightarrow}
\quad
\ListDef
\]

By refining the indices of a datatype, we make it logically more
discriminating. For example, we can ornament natural numbers to finite
sets:
\[
\NatDef 
\quad
\stackrel{\FinO}{\Rightarrow}
\quad
\FinDef
\]


\subsection{Ornaments}


We recall the definition of the universe of ornaments in
\figurename~\ref{fig:universe-orn}. Besides our ability to copy the
original description (with the codes \(\DUnit\), \(\DSigma\), and
\(\DPi\)), we can \(\OInsert\) new \(\Sigma\)-types, \(\ODelete\)
\(\Sigma\)-types by providing a witness, and use a more precise index
in the \(\DVar\) codes. While this universe is defined on
\(\IDesc[K]\), \ie functors from \(\Slice{\Set}{K}\) to \(\Set\), it
readily lifts to endofunctors on slices, \ie on descriptions
\(\idesc[K\: L]\):
\smallskip\\
\[\Code{
\ornDef 
\medskip \\
\InterpretornDef
}\]

\renewcommand{\ornDef}{
\Code{
  \Let{\orn\:
       \PiTel{\Var{D}}{\idesc[\Var{K}\: \Var{L}]}\:
       \PiTel{\Var{\MorphSource}}{\Var{I} \To \Var{K}}\:
       \PiTel{\Var{\MorphTarget}}{\Var{J} \To \Var{L}}}
      {\Set[1]}
      {} 
\\
  \orn\: \Var{D}\: \Var{\MorphSource}\: \Var{\MorphTarget}
      \DoReturn
         {\PiTo{\Var{j}}{\Var{J}}
             \Orn[(\Var{D}\: (\Var{\MorphTarget}\: \Var{j}))\: \Var{\MorphSource}]}
  }
}

\renewcommand{\InterpretornDef}{
  \Code{
  \Let{\Interpretorn{\PiTel{\Var{o}}{\orn[\Var{D}\: \Var{\MorphSource}\: \Var{\MorphTarget}]}}}
      {\idesc[\Var{I}\: \Var{J}]}{}\\
  \Case{
  \Return{\Interpretorn{\Var{o}}}
         {\Lam{\Var{j}}\InterpretOrn{\Var{o}\: \Var{j}}}
  }}
}

\renewcommand{\OrnDef}{
\Data{\Orn}
     {\Index{\Var{D}}{\IDesc[\Var{K}]}
      \Param{\Var{\MorphSource}}{\Var{I} \To \Var{K}}}
     {\Set[1]}
     {
 \multicolumn{4}{l}{\quad\CommentLine{Extend with \(S\):}} \\
  \Emit{\Orn}
       {\Var{D} &
        \Var{\MorphSource}}
       {\OInsert[\PiTel{\Var{S}}{\Set}
                 \PiTel{\Var{D^+}}{\Var{S} \To \Orn[\Var{D}\: \Var{\MorphSource}]}]}
  \multicolumn{4}{l}{\quad\CommentLine{Refine index:}} \\
  \Emit{\Orn}
       {(\DVar[\Var{k}]) &
        \Var{\MorphSource}}
       {\DVar[\PiTel{\Var{i}}{\Var{\MorphSource} \Inverse \Var{k}}]}
  \multicolumn{4}{l}{\quad\CommentLine{Copy the original:}} \\
  \Emit{\Orn}
       {\DUnit &
        \Var{\MorphSource}}
       {\DUnit}
  \Emit{\Orn}
       {(\DPi[\Var{S}\: \Var{T}]) &
        \Var{\MorphSource}}
       {\DPi[\PiTel{\Var{T^+}}
                   {\PiTo{\Var{s}}{\Var{S}}
                         {\Orn[(\Var{T}\: \Var{s})\: \Var{\MorphSource}]}}]}
  \Emit{\Orn}
       {(\DSigma[\Var{S}\: \Var{T}]) &
        \Var{\MorphSource}}
       {\DSigma[\PiTel{\Var{T^+}}
                   {\PiTo{\Var{s}}{\Var{S}}
                         {\Orn[(\Var{T}\: \Var{s})\: \Var{\MorphSource}]}}]}
  \multicolumn{4}{l}{\quad\CommentLine{Delete \(\DSigma[S]\):}} \\
& \OrEmit{\ODelete[\PiTel{\Var{s}}{\Var{S}}
                   \PiTel{\Var{T^+}}{\Orn[(\Var{T}\: \Var{s})\: \Var{\MorphSource}]}]}}}

\renewcommand{\InterpretOrn}[1]{\Interpretorn{#1}}

\renewcommand{\InterpretOrnDef}{
  \Let{\InterpretOrn{\PiTel{\Var{O}}{\Orn[\Var{D}\: \Var{\MorphSource}]}}}
      {\IDesc[\Var{I}]}{
  \Return{\InterpretOrn{\OInsert[\Var{S}\: \Var{D^+}]}}
         {\DSigma[\Var{S}\: \Lam{\Var{s}}{\InterpretOrn{\Var{D^+}\: \Var{s}}}]}
  \Return{\InterpretOrn{\DVar[(\Inv[\Var{i}])]}}
         {\DVar[\Var{i}]} 
  \Return{\InterpretOrn{\DUnit}}
         {\DUnit}
  \Return{\InterpretOrn{\DPi[\Var{T^+}]}}
         {\DPi[\Var{S}\: \Lam{\Var{s}}{\InterpretOrn{\Var{T^+}\: \Var{s}}}]}
  \Return{\InterpretOrn{\DSigma[\Var{T^+}]}}
         {\DSigma[\Var{S}\: \Lam{\Var{s}}{\InterpretOrn{\Var{T^+}\: \Var{s}}}]}
  \Return{\InterpretOrn{\ODelete[\Var{s}\: \Var{T^+}]}}
         {\InterpretOrn{\Var{T^+}\: \Var{s}}}
 }
}

\begin{figure}

\centering
\subfloat[][Code]{\footnotesize\(\OrnDef\)} \qquad
\subfloat[][Interpretation]{\footnotesize\(\InterpretOrnDef\)}

\caption{Universe of ornaments}
\label{fig:universe-orn}

\end{figure}

\begin{example}[Ornamenting natural numbers to list]
  \label{ex:nat-orn-list}

We obtain list from natural numbers with the following ornament:
\[
\Let{\ListO & \PiTel{\Var{A}}{\Set}}
    {\orn[\NatD\: \Identity\: \Identity]}
    {
\Return{\ListO & \Var{A}}
       {\Lam{\Void}
         \DSigma\: \LambdaBinder
           \left\{
           \begin{array}{@{}l@{\DoReturn}l}
            \InjLeft[\Void] & \DUnit \\
            \InjRight[\Void] & \OInsert[\Var{A}\: \Lam{\_}{\DVar[\Void]}]
           \end{array}
           \right.}
}
\]
The reader will check that the interpretation (\(\Interpretorn{\_}\))
of this ornament followed by the interpretation
(\(\InterpretIDesc{\_}\)) of the resulting description yields the
signature functor of list \(X \mapsto 1 + A \times X\).

\end{example}
\begin{example}[Ornamenting natural numbers to finite sets]
  \label{ex:nat-orn-fin}

We obtain finite sets by inserting a number \(\TypeAnn{n'}{\Nat}\),
constraining the index \(n\) to \(\Suc[n']\), and -- in the recursive
case -- indexing at \(n'\):
\[
\Let{\FinO}
    {\orn[\NatD\: (\Lam{\Var{n}}{\Void})\: (\Lam{\Var{n}}{\Void})]}
    {
\Return{\FinO}
       {\Lam{\Var{n}}
         \begin{array}[t]{@{}l}
         \OInsert\: \Nat\:
                    \Lam{\Var{n'}} 
         \OInsert\: (\Var{n} \PropEqual \Suc[\Var{n'}])\:
                    \Lam{\_} \\
         \DSigma\: \LambdaBinder
           \left\{
           \begin{array}{@{}l@{\DoReturn}l}
            \InjLeft[\Void] & \DUnit \\
            \InjRight[\Void] & \DVar[\Var{n'}]
           \end{array}
           \right.
         \end{array}
       }
}
\]
Again, the reader will verify that this is indeed describing the
signature of finite sets.

\end{example}
A detailed account of ornaments from a programmer's perspective will
be found
elsewhere~\citep{mcbride:ornament,dagand:fun-orn,ko:modularising-inductive}. For
the purpose of this paper, these definitions are enough. We shall
refer to the aforementioned papers when programming concepts reappear
in our categorical framework.

\subsection{Ornaments are cartesian morphisms}

\newcommand{\CatOrnExt}{\mathit{ext}}
\newcommand{\CatOrnNext}{\Next^+}
\newcommand{\CatOrnCoe}{q}

\newcommand{\aExt}{e}

Relating the definition of ornaments with our polynomial reading of
descriptions, we make the following remarks. Firstly, the ornament
code lets us only insert -- with the \(\OInsert\) code -- or delete --
with the \(\ODelete\) code -- \(\DSigma\) codes while forbidding
deletion or insertion of either \(\DPi\) or \(\DVar\) codes. In terms
of container, this translates to: shapes can be extended, while
positions must be isomorphic. Secondly, on the \(\DVar\) code, the
ornament code lets us pick any index in the inverse image of \(u\). In
terms of container, this corresponds to the coherence condition: the
initial indexing must commute with applying the ornamented indexing
followed by \(u\).  Concretely, for a container
\(\IContainer{\Shape}{\Position}{\Next}\), an ornament can be modelled
as an extension \(\CatOrnExt\), a refined indexing \(\CatOrnNext\)
subject to coherence condition \(\CatOrnCoe\) with respect to the
original indexing:
\[
\left\{
\Code{
\TypeAnn{\Var{\CatOrnExt}}{\Var{\Shape}\: (\Var{\MorphTarget}\: \Var{l}) \To \Set} \\
\TypeAnn{\Var{\CatOrnNext}}{\Var{\CatOrnExt}\:\Var{\aShape} \To \Var{\Position}\: \Var{\aShape} \To \Var{K}} \\
\TypeAnn{\Var{\CatOrnCoe}}{\Forall{\Var{\aPosition}}{\Var{\Position}\: \Var{\aShape}}
                  \Var{\MorphSource}\: (\Var{\Next^+}\: \Var{\aExt}\: \Var{\aPosition}) \PropEqual \Var{\Next}\:\Var{\aPosition}}
}
\right.
\]
%
%
Equivalently, the family of set \(\CatOrnExt\) can be understood as
the inverse image of a function
\(\TypeAnn{\MorphShape}{\Var{\Shape^+}\:\Var{l} \To \Var{\Shape}\:
  (\Var{\MorphTarget}\: \Var{l})}\). The function \(\CatOrnNext\) is
then the next index function of a container with shapes \(\Shape^+\)
and positions \(\Position \Compose \MorphShape\). Put otherwise, the
morphism on shapes \(\MorphShape\) together with the coherence
condition \(\CatOrnCoe\) form a cartesian morphism from
\(\IContainer{\Shape^+}{\Position \Compose \MorphShape}{\Next^+}\) to
\(\IContainer{\Shape}{\Position}{\Next}\)! To gain some intuition, the
reader can revisit the cartesian morphism of
Example~\ref{example:cart-morph-list-nat} as an ornament of container
-- by simply inverting the morphism on shapes -- and as an ornament of
description -- by relating it with the ornament \(\ListO\)
(Example~\ref{ex:nat-orn-list}).

We shall now formalise this intuition by proving the following
isomorphism:
\begin{lemma}
  \label{lemma:orn-to-polyc}

Ornaments describe cartesian morphisms between polynomial functors, \ie we have the isomorphism
\[
\orn[D\: u\: v] \Iso \PolyC{\LCCC{E}}(\_,D)_{u,v}
\]

\end{lemma}


In terms of cartesian morphism of polynomials, extending the shape
corresponds to the morphism \(\alpha\). Enforcing that the positions,
\ie the structure, of the datatype remain the same corresponds to the
pullback along \(\alpha\). The refinement of indices corresponds to
the frame morphisms commuting.

\newcommand{\OrnToCartArr}{\green{\phi}}
\newcommand{\CartArrToOrn}{\green{\psi}}

\begin{proof}


We develop the proof on the container presentation: this lets us
work in type theory, where is anchored the definition of ornaments. It
is a necessary hardship, since no other decent model of ornaments is
available to us. After this bootstrapping process, we shall have the
abstract tools necessary to lift off type theory.


\Spacedcommand{\OrnToCartForget}{\Function{forget}}

The first half of the isomorphism consists of mapping an ornament
\(o\) of a description \(D\) to a cartesian morphism from the
container described by \(\Interpretorn{o}\) to the container described
by D. By definition of cartesian morphisms, we simply have to give a
map from the shape of \(\Interpretorn{o}\) to the shape of \(D\):
\[\Code{
\Let{\OrnToCartArr & 
     \PiTel{\Var{o}}{\orn[\Var{D}\: \Var{\MorphSource}\: \Var{\MorphTarget}]}}
    {\IDescToIS{\Interpretorn{\Var{o}}}
          \IContainerCartArr{\Var{\MorphSource}}{\Var{\MorphTarget}}
          \IDescToIS{\Var{D}}}{
\Return{\OrnToCartArr & \Var{o}}
       {\IContainerCartMor{(\Lam{\Var{i}}{\OrnToCartForget[\Var{o}\: (\Var{\MorphSource}\: \Var{i})]})}
                          {\_}{\_}\: \where}
} \\
\quad
\Let{\OrnToCartForget &
     \PiTel{\Var{O}}{\Orn[\Var{D}\: \Var{\MorphSource}]} &
     \PiTel{\Var{sh}}{\IDescToShape[\InterpretOrn{\Var{O}}]}}
    {\IDescToShape[\Var{D}]}{
\Return{\OrnToCartForget & 
        \DUnit & 
        \Void}
       {\Void}
\Return{\OrnToCartForget & 
        (\DPi[\Var{T^+}]) &
        \Var{f}}
       {\Lam{\Var{a}}{\OrnToCartForget[(\Var{T^+}\: \Var{a})\: (\Var{f}\: \Var{a})]}}
\Return{\OrnToCartForget & 
        (\DSigma[\Var{T^+}]) &
        \Pair{\Var{a}}{\Var{sh}}}
       {\Pair{\Var{a}}{\OrnToCartForget[(\Var{T^+}\: \Var{a})\: \Var{sh}]}}
\Return{\OrnToCartForget &
        (\DVar[(\Inv[\Var{j}])]) &
        \Void}
       {\Void} 
\Return{\OrnToCartForget & 
        (\OInsert[\Var{a}\: \Var{D^+}]) &
        \Pair{\Var{a}}{\Var{sh}}}
       {\OrnToCartForget[(\Var{D^+}\: \Var{a})\: \Var{sh}]}
\Return{\OrnToCartForget & 
        (\ODelete[\Var{s}\: \Var{O}]) &
        \Var{sh}}
       {\Pair{\Var{s}}{\OrnToCartForget[\Var{O}\: \Var{sh}]}}
}
}\]


We are then left to check that (extensionally) the positions are
constructed by pullback and the indexing is coherent. This is indeed
the case, even though proving it in type theory is cumbersome. On
positions, the ornament does not introduce or delete any new \(\DPi\)
or \(\DVar\): hence the positions are left unchanged. On indices, we
rely on \(u \Inverse k\) to ensure that the more precise indexing is
coherent by construction.


In the other direction, we are given a cartesian morphism from
\(\Var{F}\) to \(\Var{G}\). We return an ornament of the description
of \(G\). For the isomorphism to hold, this ornament must interpret to
the description of \(F\):
\[
\Let{\CartArrToOrn & 
     \PiTel{\Var{m}}{\Var{F} \IContainerCartArr{\Var{\MorphSource}}{\Var{\MorphTarget}} \Var{G}}}
    {\orn[\ISToIDesc{\Var{G}}\: \Var{\MorphSource}\: \Var{\MorphTarget}]}{
\Return{\CartArrToOrn & 
        (\IContainerCartMor{\Var{\CN{forget}}}{\_}{\_})}
       {\Lam{\Var{j}}
            {\DSigma[\Lam{\Var{sh}}
             \OInsert[(\Var{\CN{forget}} \Inverse \Var{sh})\: \Lam{\Var{ext}}
             \DPi[\Lam{\Var{ps}} 
             \DVar[(\Inv[(\Var{n}_{\Var{F}}\: \Var{ps})])]]]]}}
}
\]


Indeed, the description of \(G\) is a \(\DSigma\) of its shape,
followed by a \(\DPi\) of its positions, terminated by a \(\DVar\) at
the next index. To ornament \(G\) to \(F\), we simply have to
\(\OInsert\) the inverse image of \(\Var{\CN{forget}}\), \ie the
information that extends \(G\) to \(F\). As for the next index, we
can legitimately use \(F\)'s indexing function: the coherence
condition of the cartesian morphism ensures that it is indeed in the
inverse image of the reindexing function.


Having carefully crafted the definition of \(\OrnToCartArr\) and
\(\CartArrToOrn\), it should be obvious that these functions are
inverse of each other. It is sadly not that obvious to an
(intensional) theorem prover. Hence, we will not attempt to prove it
in type theory here.

\end{proof}


\paragraph{Relation with ornamental algebras~\citep[\S{4}]{mcbride:ornament}} 
To introduce the notion of ``ornamental algebra'', the second author
implemented the \(\Function{erase}\) helper function taking an
ornamented type to its unornamented form. This actually corresponds to
our transformation \(\OrnToCartArr\), followed by the interpretation
of the resulting cartesian morphism. The \(\Function{erase}\) function
given in the original presentation is indeed natural and cartesian. 


\newcommand{\IDescC}{\mathit{IDesc}^c}

In the previous section, we have established a connection between
descriptions and polynomials. We have now established a connection
between ornaments and cartesian morphisms of polynomials. It thus
makes sense to organise descriptions in a framed bicategory
\(\IDescC\):
\begin{definition}[Framed bicategory \(\IDescC\)]

  The framed bicategory \(\IDescC\) is defined by:
\begin{itemize}
\item Objects: sets
\item Vertical morphisms: set morphisms
\item Horizontal morphisms: descriptions, framed by \(I\) and \(J\)
\item Squares: a square from \(F\) to \(G\) framed by \(u\) and
  \(v\) is an ornament \(\TypeAnn{o}{\orn[G\: u\: v]}\) of \(G\) that
  interprets to (a code isomorphic to) \(F\)
\end{itemize}
Where, as for \(\PolyC{\LCCC{E}}\) (Example~\ref{example:polyc-bicat}),
the frame structure consists in reindexing a description along a pair
of functions.

\end{definition}


\subsection{A framed biequivalence}

We are now ready to establish an equivalence of category between
\(\IDescC\) and \(\PolyFuncC{\LCCC{E}}\), thus completing our journey
from the type theoretical definition of ornaments to its model as
cartesian morphisms.

\begin{proposition}
  The double category \(\IDescC\) is framed biequivalent to
  \(\PolyFuncC{\LCCC{E}}\).
\end{proposition}



\begin{proof}


As for the proof of Lemma~\ref{lemma:orn-to-polyc}, we work from
\(\IDescC\) to \(\ICont\) to prove this theorem. Since \(\ICont\) is
equivalent to \(\PolyFuncC{\LCCC{E}}\), this gives the desired result. To
prove a framed biequivalence, we need a functor on the base category
and another on the total category. In this particular case, both base
categories are \(\Set\): we shall therefore take the identity functor,
hence trivialising the natural isomorphisms on composition,
identity, and frames.


On the total category, we prove the equivalence by exhibiting a full
and faithful functor from \(\IDesc\) to \(\ICont\) that is essentially
surjective on objects. Unsurprisingly, this functor is defined on
objects by \(\IDescToIS{\_}\), which is indeed essentially
surjective by Lemma~\ref{lemma:idesc-to-is}. The morphism part is
defined by \(\OrnToCartArr\), which is full and faithful by
Lemma~\ref{lemma:orn-to-polyc}.

\end{proof}


We have therefore established the following equivalences of framed
bicategories:
\[
\begin{tikzpicture}
\matrix (m) [matrix of math nodes
            , row sep=1em
            , column sep=2em
            , text height=1.5ex
            , text depth=0.25ex
            , ampersand replacement=\&]
{ 
    \IDescC  \&        \& \PolyC{\LCCC{E}}  \& \PolyFuncC{\LCCC{E}} \\
             \& \ICont \&               \&                  \\
};
\path[<->] 
   (m-1-1) edge (m-2-2)
   (m-2-2) edge (m-1-3)
   (m-1-3) edge (m-1-4)
;
\path[dashed,<->]
   (m-1-1) edge (m-1-3)
;
\end{tikzpicture}
\]
%
%
We may now conflate the notions of ornament, cartesian morphism, and
cartesian natural transformation. In particular, we shall say that
``\(F\) ornaments \(G\)'' when we have a cartesian morphism from \(F\) to
\(G\). Let us now raid the polynomial toolbox for the purpose of
programming with ornaments. The next section shows the beginning of
what is possible.


\section{Tapping into the categorical structure}
\label{sec:orn-structure}


In the previous section, we have characterised the notion of ornament
in terms of cartesian morphism. We now turn to the original ornamental
constructions~\citep{mcbride:ornament} -- such as the ornamental algebra
and the algebraic ornament -- and rephrase them in our categorical
framework. Doing so, we extract the structure governing their type
theoretic definition.

Next, we study the categorical structure of cartesian morphisms and
uncover novel and interesting ornamental constructions. We shall see
how the identity, composition, and frame reindexing translate into
ornaments. We shall also be interested in pullbacks in the category
\(\PolyFuncC{\LCCC{E}}\) and the functoriality of the derivative in
that category.


\subsection{Ornamental algebra}

\newcommand{\CartesianArr}[2]{\overset{#1}{\underset{#2}{\Longrightarrow^c}}}

Ornamenting a datatype is an effective recipe to augment it with new
information. We thus expects that, given an ornamented object, we can
\emph{forget} its extra information and regain a raw object. This
projection is actually a generic operation, provided by the
\emph{ornamental algebra}. It is a corollary of the very definition of
ornaments as cartesian morphisms.

\renewcommand{\OrnForgetAlg}[1]{\Function{\(#1\)-forgetAlg}}

\begin{corollary}[Ornamental algebra]
From an ornament \(\TypeAnn{o}{F \CartesianArr{\MorphSource}{\MorphTarget} G}\), we obtain the
\emph{ornamental algebra} \(\TypeAnn{\OrnForgetAlg{o}}{F\: (\lfp{G} \circ v) \to
  \lfp{G} \circ u}\).
\end{corollary}

\begin{proof}


  We apply the natural transformation \(o\) at \(\lfp{G}\) and
  post-compose by \(\inLfp\):
\[
\OrnForgetAlg{o} : F\: (\lfp{G} \circ v) \stackrel{o_{\lfp{G}}}{\longrightarrow} 
                   (G\: \lfp{G}) \circ u \stackrel{\inLfp}{\longrightarrow}
                   \lfp{G} \circ u
\]
\end{proof}


Folding the ornamental algebra, we obtain a map from the ornamented
type \(\lfp{F}\) to its unornamented version \(\lfp{G}\). In effect,
the ornamental algebra describes how to \emph{forget} the
extra-information introduced by the ornament.


\begin{example}[Ornamental algebra of the \(\List{}\) ornament]
  \label{ex:orn-alg-list}

The cartesian morphism from list to natural numbers
(Example~\ref{example:cart-morph-list-nat}) maps the \(\Nil\)
constructor to \(\Zero\), while the \(\Cons\) constructor is mapped to
\(\Suc\). Post-composing by \(\inLfp\), we obtain a natural
number. This is the algebra computing the length of a list.

\end{example}


\subsection{Algebraic ornaments}

\newcommand{\AlgOrn}[2]{{#1}^{#2}}


The notion of algebraic ornament was initially introduced by the
second author~\citep{mcbride:ornament}. A similar categorical
construction, defined for any functor, was also presented by~\citet{atkey:inductive-refinement}. In this section, we reconcile
these two works and show that, for a polynomial functor, the
refinement functor can itself be internalised as a polynomial functor.


\begin{definition}[Refinement functor~{\citep[\S{4.3}]{atkey:inductive-refinement}}]

Let \(F\) an endofunctor on \(\Slice{\LCCC{E}}{I}\). Let
\((\TypeAnn{X}{\Slice{\LCCC{E}}{I}}, \TypeAnn{\alpha}{F\: X \to X})\)
an algebra over \(F\).

The \emph{refinement functor} is defined by:
\[
\AlgOrn{F}{\alpha} \triangleq 
    \TypeAnn{\CatExists{\alpha} \circ \Lift{F}}
            {\Slice{(\Slice{\LCCC{E}}{I})}{X} \to \Slice{(\Slice{\LCCC{E}}{I})}{X}}
\]
Where \(\Lift{F}\) -- the lifting of
\(F\)~\citep{hermida:induction,fumex:phd} -- is taken, in an LCCC, to
be the morphism part of the functor \(F\).

\end{definition}


The idea, drawn from refinement types~\citep{freeman:refinement}, is
that a function \(\TypeAnn{\Fold{\alpha}}{\lfp{F} \To X}\) can be
thought of as a predicate over \(\lfp{F}\). By \emph{integrating} the
algebra \(\alpha\) \emph{into} the signature \(F\), we obtain a
signature \(\AlgOrn{F}{\alpha}\) indexed by \(X\) that describes the
\(F\)-objects satisfying, by construction, the predicate
\(\Fold{\alpha}\). Categorically, this translates to:
\begin{theorem}[Coherence property of algebraic ornament]
The fixpoint of the algebraic ornament of \(P_F\) by \(\alpha\)
satisfies the isomorphism
\(
\lfp{\AlgOrn{P_F}{\alpha}} \Iso \Sigma_{\Fold{\alpha}} \One[\lfp{F}]
\)
where \(\TypeAnn{\One}{\SliceL{E}{I} \to [\SliceL{E}{I},
    \SliceL{E}{I}]}\), the terminal object functor, maps objects \(X\)
to \(\CatId[X]\). 
\end{theorem}
%
%
\begin{proof}
This is an application of Theorem
4.6~\citep{atkey:inductive-refinement}, specialised to the codomain
fibration (\ie an LCCC).
\end{proof}


Informally, using a set theoretic notation, this isomorphisms reads
as
\(
\lfp{\AlgOrn{F}{\alpha}}\: i\: x \Iso 
    \ISet{\TypeAnn{t}{\lfp{F}\: i}}{\Fold{\alpha}\: t \PropEqual x}
\).
That is, the algebraic ornament \(\lfp{\AlgOrn{F}{\alpha}}\) at index
\(i\) and \(x\) corresponds \emph{exactly} to the pair of a witness
\(t\) of \(\lfp{F}\: i\) and a proof that this witness satisfies the
indexing equation \(\Fold{\alpha}\: t \PropEqual x\).  In effect, from
an algebraic predicate over an inductive type, we have an effective
procedure reifying this predicate as an inductive family. This theorem
also has an interesting computational interpretation. Crossing the
isomorphism from left to right, we obtain the
\(\Function{Recomputation}\) theorem\cite[\S{8}]{mcbride:ornament}:
from any \(\TypeAnn{t^+}{\lfp{\AlgOrn{F}{\alpha}}\: i\: x}\), we can
extract a \(\TypeAnn{t}{\lfp{F}\: i}\) together with a proof that
\(\Fold{\alpha}\: t\) equals \(x\). From right to left, we obtain the
\(\Function{remember}\) function~\citep[\S{7}]{mcbride:ornament}: from
any \(\TypeAnn{t}{\lfp{F}\: i}\), we can lift it to its ornamented
form with \(\TypeAnn{\Function{remember}\:
  t}{\lfp{\AlgOrn{F}{\alpha}}\: i\: (\Fold{\alpha}\: t)}\).


When \(F\) is a polynomial functor, we show that the refinement
functor can be internalised and presented as an ornament of \(F\). In
practice, this means that from a description \(D\) and an algebra
\(\alpha\), we can \emph{compute} an ornament code that describes the
functor \(\AlgOrn{D}{\alpha}\). This should not come as a surprise:
algebraic ornaments were originally presented as ornamentations of the
initial description~\citep[\S{5}]{mcbride:ornament}. The following
theorem abstracts the original definition:
\begin{proposition}

  Let \(F\) a polynomial endofunctor on \(\Slice{\LCCC{E}}{I}\). Let
  \((X,\alpha)\) an algebra over \(P_F\), \ie \(\TypeAnn{\alpha}{P_F X
    \to X}\).
  The refinement functor \(\AlgOrn{P_F}{\alpha}\) is polynomial and ornaments \(F\).

\end{proposition}

\begin{proof}

To show that \(\AlgOrn{P_F}{\alpha}\) is a polymonial ornamenting
\(F\), we exhibit a cartesian natural transformation from
\(\AlgOrn{P_F}{\alpha}\) to \(P_F\). Since \(P_F\) is polynomial, we
obtain that \(\AlgOrn{P_F}{\alpha}\) is polynomial by
Lemma~\ref{lemma:cart-induce-poly}.


First, there is a cartesian natural transformation from the lifting
\(\Lift{P_F}\) to \(P_F\). Indeed, for an LCCC, the lifting consists
of the morphism part of \(P_F\), denoted
\(\CatArr{P_F}\)~\citep{fumex:phd}. We therefore have the following
isomorphism, hence cartesian natural transformation:
\[
\begin{tikzpicture}
\matrix (m) [matrix of math nodes
            , row sep=3em
            , column sep=10em
            , text height=2ex
            , text depth=1ex
            , ampersand replacement=\&]
{ 
    \Slice{(\Slice{\LCCC{E}}{I})}{X} \&
        \Slice{(\Slice{\LCCC{E}}{I})}{P_F\: X} \\
    \Slice{\LCCC{E}}{I} \Iso \Slice{(\Slice{\LCCC{E}}{I})}{\Terminal} \&
        \Slice{(\Slice{\LCCC{E}}{I})}{\Terminal} \Iso \Slice{\LCCC{E}}{I} \\
};
\node (SetF) at ($(m-2-1) !.5! (m-2-2)$) {\(\Slice{(\Slice{\LCCC{E}}{I})}{P_F\: \Terminal}\)}
;
\path[->] 
    (m-1-1) edge node[above] {\(\CatArr{P_F}\)} (m-1-2)
    (m-2-1) edge node[below] {\(\CatArr{P_F}\)} (SetF)
    (SetF) edge node[below] {\(\CatExists{!_{P_F\: \Terminal}}\)} (m-2-2)
    (m-2-1) edge[bend right=30] node[below] {\(P_F\)} (m-2-2)
    (m-1-1) edge node[left] {\(\CatExists{!_X}\)} (m-2-1)
    (m-1-2.south west) edge node[left=10pt] {\(\CatExists{P_F\: !_X}\)} (SetF)
    (m-1-2.south) edge node[right] {\(\CatExists{!_{P_F\: X}}\)} (m-2-2)
;    
\end{tikzpicture}
\]
Indeed, unfolding the definition of \(\CatExists{f} \triangleq f \circ
\_\), the left square reduces to the functoriality of \(P_F\). The
right triangle is simply the op-cartesian lifting of 
\(
\begin{tikzpicture}[baseline=-0.65ex,scale=0.25]
\matrix (m) [matrix of math nodes
            , row sep=3em
            , column sep=6em
            , text height=1.5ex
            , text depth=0.25ex
            , ampersand replacement=\&]
{ 
                    \& P_F\: X \\
    P_F\: \Terminal \& \Terminal \\
};
\path[->] 
    (m-1-2) edge node[above=1ex] {\(P_F\: !_{X}\)} (m-2-1)
    (m-1-2) edge node[right] {\(!_{P_F X}\)} (m-2-2)
    (m-2-1) edge node[below] {\(!_{P_F \Terminal}\)} (m-2-2)
;    
\end{tikzpicture}
\). The bottom triangle commutes by the isomorphism relating the slice
over the terminal and the total category, \ie
\(\Slice{(\Slice{\LCCC{E}}{I})}{\Terminal} \Iso \Slice{\LCCC{E}}{I}\) .


There is also a cartesian natural transformation from
\(\CatExists{\alpha}\) to the identity polynomial indexed by \(J\):
\[
\begin{tikzpicture}
\matrix (m) [matrix of math nodes
            , row sep=3em
            , column sep=3em
            , text height=1.5ex
            , text depth=0.25ex
            , ampersand replacement=\&]
{ 
    \Slice{\LCCC{E}}{\CatExists{!} P_F\: X} \Iso \Slice{(\Slice{\LCCC{E}}{I})}{P_F\: X} \&
    \Slice{(\Slice{\LCCC{E}}{I})}{X} \Iso \Slice{\LCCC{E}}{\CatExists{!} X} \\
    \Slice{\LCCC{E}}{I} \Iso \Slice{(\Slice{\LCCC{E}}{I})}{\Terminal} \&
    \Slice{(\Slice{\LCCC{E}}{I})}{\Terminal} \Iso \Slice{\LCCC{E}}{I}  \\
};
\path[->] 
   (m-1-1) edge node[above] { \(\CatExists{\alpha}\) } (m-1-2)
   (m-2-1) edge[double,-] (m-2-2) 
   (m-1-1) edge node[left] { \(\CatExists{!_{P_F X}}\) } (m-2-1)
   (m-1-2) edge node[right] { \(\CatExists{!_X}\) } (m-2-2)
;
\node at ($(m-1-1) !.5! (m-2-2)$) {\(\Iso\)} ;
\end{tikzpicture}
\]


By horizontal composition of these two cartesian natural
transformations, we obtain a cartesian natural transformation from
\(\AlgOrn{P_F}{\alpha}\) to \(\CatId \circ P_F \Iso P_F\).

\end{proof}

\begin{remark}

  This proof is not entirely satisfactory: it is specialised to the
  predicate lifting in the codomain fibration. The construction of the
  cartesian natural transformation from the lifting to the functor is
  therefore a rather pedantic construction. Hopefully, a more abstract
  proof could be found.

\end{remark}


\subsection{Categorical structures}


\paragraph{Identity}
A trivial ornamental construction is the \emph{identity}
ornament. Indeed, for every polynomial, there is a cartesian morphism
from and to itself, introducing no extension and no refinement. In
terms of \(\Orn\) code, this construction simply consists in copying
the code of the description: this is a generic
program, taking a description as input and returning the identity
ornament.


\paragraph{Vertical composition}
The next structure of interest is composition. Recall that an ornament
corresponds to a (cartesian) natural transformation. There are
therefore two notions of composition. First, vertical composition lets
us collapse chains of ornaments:
\[
\begin{tikzpicture}[description/.style={fill=white,inner sep=2.5pt}]
\matrix (m) [matrix of math nodes
            , row sep=5em
            , column sep=6em
            , text height=1.5ex
            , text depth=0.25ex
            , ampersand replacement=\&]
{ 
    \Slice{\LCCC{E}}{I} \& \Slice{\LCCC{E}}{J} 
    \& \Slice{\LCCC{E}}{I} \& \Slice{\LCCC{E}}{J} \\
};
\path[->] 
   (m-1-1) edge[bend left=60] node[above] { \(F\) } (m-1-2)
   (m-1-1) edge node[description] { \(G\) } (m-1-2)
   (m-1-1) edge[bend right=60] node[below] { \(H\) } (m-1-2)
   (m-1-3) edge[bend left=60] node[above] { \(F\) } (m-1-4)
   (m-1-3) edge[bend right=60] node[below] { \(H\) } (m-1-4)
;
\path 
   (m-1-1) -- (m-1-2) node [midway,above=4pt] {\(\quad\Downarrow o_1\)}
   (m-1-1) -- (m-1-2) node [midway,below=4pt] {\(\quad\Downarrow o_2\)}
   (m-1-3) -- (m-1-4) node [midway] {\(\Code[c]{o_2 \bullet o_1\\\Downarrow}\)}
;
\node at ($(m-1-2) !.5! (m-1-3)$) {\(\Iso\)} ;
\end{tikzpicture}
\]

\begin{example}[Vertical composition of ornaments]

We have seen that \(\List{}\) ornaments \(\Nat\). We also know that
\(\Vector{}\) ornaments \(\List{}\). By vertical composition, we thus
obtain that \(\Vector{}\) ornaments \(\Nat\).

\end{example}


\paragraph{Horizontal composition}
Turning to horizontal composition, we have the following identity:
\[
\begin{tikzpicture}
\matrix (m) [matrix of math nodes
            , row sep=3em
            , column sep=6em
            , text height=1.5ex
            , text depth=0.25ex
            , ampersand replacement=\&]
{ 
    \Slice{\LCCC{E}}{I} \& \Slice{\LCCC{E}}{J} \& \Slice{\LCCC{E}}{K}  \\ 
    \Slice{\LCCC{E}}{I} \&                     \& \Slice{\LCCC{E}}{K} \\
};
\path[->] 
   (m-1-1) edge[bend left=30] node[above] { \(F_1\) } (m-1-2)
   (m-1-1) edge[bend right=30] node[below] { \(G_1\) } (m-1-2)
   (m-1-2) edge[bend left=30] node[above] { \(F_2\) } (m-1-3)
   (m-1-2) edge[bend right=30] node[below] { \(G_2\) } (m-1-3)
   (m-2-1) edge[bend left=20] node[above] { \(F_2 \circ F_1\) } (m-2-3)
   (m-2-1) edge[bend right=20] node[below] { \(G_2 \circ G_1\) } (m-2-3)
;
\node at ($(m-1-1) !.55! (m-1-2)$) {\(\Downarrow o_1\)} ;
\node at ($(m-1-2) !.55! (m-1-3)$) {\(\Downarrow o_2\)} ;
\node at ($(m-2-1) !.5! (m-2-3)$) {\(\Code[c]{o_2 \circ o_1\\\Downarrow}\)} ;
\node at ($(m-1-1) !.5! (m-2-1)$) {\(\Iso\)};
\end{tikzpicture}
\]


\newcommand{\PList}{\CN{PList}_A}
\newcommand{\PVec}{\CN{PVec}_A}
\newcommand{\PHeight}{\Function{Height}}
\newcommand{\PSquare}{\Function{Square}}

\begin{example}[Horizontal composition of ornaments]
  \label{ex:bbin-orn-bin-tree}

Let us consider the following polynomials:
{
\[
\begin{array}{l}
\TypeAnn{\PSquare\: X \mapsto X \times X}
        {\Slice{\Set}{\Unit} \to \Slice{\Set}{\Unit}} \\
\TypeAnn{\PHeight\: \ISet{X_n}{n \in \Nat} \mapsto \ISet{X_n \times X_{n+1}}{n \in \Nat} \\
                  \qquad\qquad\qquad + \ISet{X_n \times X_n}{n \in \Nat}}
        {\Slice{\Set}{\Nat} \to \Slice{\Set}{\Nat}} 
\end{array}
\]}
It is easy to check that \(\VecICont\) ornaments \(\ListICont\) and
\(\PHeight\) ornaments \(\PSquare\). By horizontal composition of
these ornaments, we obtain that \(\VecICont \circ \PHeight\) --
describing a balanced binary tree -- is an ornament of \(\ListICont
\circ \PSquare\) -- describing a binary tree. Thus, we obtain that
balanced binary trees ornament binary trees.

\end{example}


\paragraph{Frame structure}
Finally, the frame structure of the bicategory lets us lift morphisms
on indices to polynomials. 

\newcommand{\Even}{\Canonical{Even}}

\begin{example}[Reindexing ornament]

Let \(\TypeAnn{\Function{twice}}{\Nat \to \Even}\), the function that
multiplies its input by 2.
The \(\Vector{}\) polynomial is indexed by \(\Nat\): we can therefore
reindex it with \(\Function{twice}\). We automatically obtain an
ornament of \(\Vector{}\) that is indexed by \(\Even\). Needless to
say, this construction is not very interesting on its own. However, in
a larger development, we can imagine retrofitting an indexed datatype
to use another index, making it usable by a library function.

\end{example}


The identity, vertical, and horizontal compositions illustrate the
algebraic properties of ornaments. The categorical simplicity of
cartesian morphisms gives us a finer understanding of datatypes and
their relation to each other, as illustrated by
Example~\ref{ex:bbin-orn-bin-tree}.


\subsection{Pullback of ornaments}


\newcommand{\PolyFuncCart}[2]{\mathit{PolyFun}^C(#1,#2)}

So far, we have merely exploited the fact that
\(\PolyFuncC{\LCCC{E}}\) is a framed bicategory. However, it has a
much richer structure. That extra structure can in turn be translated
into ornamental constructions. We shall focus on pullbacks, but we
expect other categorical notions to be of programming interest.

\begin{proposition}

The category \(\PolyFuncC{\LCCC{E}}\) has all pullbacks.

\end{proposition}


\begin{proof}


First, let us recall that the notion of \emph{cartesian} morphism
arises from the fact that the following functor is a fibration:
\[
\begin{tikzpicture}
\matrix (m) [matrix of math nodes
            , row sep=2em
            , column sep=5em
            , text height=1.5ex
            , text depth=0.25ex
            , ampersand replacement=\&]
{ 
  [\Slice{\LCCC{E}}{I}, \Slice{\LCCC{E}}{J}] \\
  \Slice{\LCCC{E}}{J} \\
};
\path[->] 
   (m-1-1) edge node[right] {\(\_\: \Terminal\)} (m-2-1)
;
\end{tikzpicture}
\]
Where cartesian natural transformation corresponds to the cartesian
morphisms of that fibration.

Let \(\TypeAnn{\phi}{F \NatTrans^c H}\) and \(\TypeAnn{\psi}{G
  \NatTrans^c H}\) two cartesian natural transformation. They are
projected to \(\phi_{\Terminal}\) and \(\psi_{\Terminal}\) in the base
category. Since \(\Slice{\LCCC{E}}{J}\) is pullback complete, we can
construct the pullback of \(\phi_{\Terminal}\) and
\(\psi_{\Terminal}\), thus obtaining the following pullback square:
\[
\begin{tikzpicture}
\matrix (m) [matrix of math nodes
            , row sep=3em
            , column sep=3em
            , text height=1.5ex
            , text depth=0.25ex
            , ampersand replacement=\&]
{ 
  \cdot            \& F \Terminal \\
  G \Terminal      \& H \Terminal \\
};
\path[->] 
   (m-1-2) edge node[right] {\(\phi_{\Terminal}\)} (m-2-2) 
   (m-2-1) edge node[below] {\(\psi_{\Terminal}\)} (m-2-2) 
   (m-1-1) edge node[above] {\(\psi^\dagger_{\Terminal}\)} (m-1-2)
   (m-1-1) edge node[left] {\(\phi^\dagger_{\Terminal}\)} (m-2-1)
;
\begin{scope}[shift=($(m-1-1)!.25!(m-2-2)$)]
\draw +(-.25,0) -- +(0,0)  -- +(0,.25);
\end{scope}
\end{tikzpicture}
\]
By reindexing, we thus obtain the following square in the total
category:
\[
\begin{tikzpicture}
\matrix (m) [matrix of math nodes
            , row sep=3em
            , column sep=3em
            , text height=1.5ex
            , text depth=0.25ex
            , ampersand replacement=\&]
{ 
  \cdot            \& F X \\
  G X      \& H X \\
};
\path[->] 
   (m-1-2) edge node[right] {\(\phi_X\)} (m-2-2) 
   (m-2-1) edge node[below] {\(\psi_X\)} (m-2-2) 
   (m-1-1) edge node[above] {\({\psi^\dagger}_X\)} (m-1-2)
   (m-1-1) edge node[left] {\({\phi^\dagger}_X\)} (m-2-1)
;
\end{tikzpicture}
\]
By Exercise 1.4.4~\citep{jacobs:categorical-logic}, we have that this
square is actually a pullback. In a nutshell, we rely on the unicity
of cartesian morphisms in the total category to prove the universal
property of pullbacks for that square.

\end{proof}


\newcommandx{\BoundList}[2][2=\!]{\Canonical{BoundedList}\ifthenelse{\isempty{#1}}{}{_{#1}}\xspace\, #2}

\begin{example}[Pullback of ornament]

Natural numbers can be ornamented to lists
(Example~\ref{ex:nat-orn-list}) as well as finite sets
(Example~\ref{ex:nat-orn-fin}). Taking the pullback of these two
ornaments, we obtain bounded lists that correspond to lists of bounded
length, with the bound given by an index \(\TypeAnn{n}{\Nat}\). Put
explicitly, the object thus computed is the following datatype:
\[
\Data{\BoundList{}}
     {\Param{\Var{A}}{\Set}
      \Index{\Var{n}}{\Nat}}
     {\Set}{
  \Emit{\BoundList{\Var{A}}[\Constraint{\Var{n}}{\Suc[\Var{n'}]}]}
       {}
       {\Nil\: \PiTel{\Var{n'}}{\Nat}} 
  \OrEmit{\Cons[\PiTel{\Var{n'}}{\Nat}
                \PiTel{\Var{a}}{\Var{A}}
                \PiTel{\Var{as}}{\BoundList{\Var{A}}[\Var{n'}]}]}}
\]

\end{example}


The pullback construction is another algebraic property of ornaments:
given two ornaments, both describing an extension of the same datatype
(\eg extending natural numbers to lists and extending natural numbers
to finite sets), we can ``merge'' them into one having both
characteristics (\ie bounded lists). In type theory,
\citet{ko:modularising-inductive} have experimented with a similar
construction for composing indexing disciplines.


\subsection{Derivative of ornament}


\newcommand{\Derive}[1]{\partial_{#1}}
\newcommand{\PolyDelta}[2]{{\Poly{\!}}^{\Derive{#2}}_{#1}}

\citet{abbott:derivative} have shown that the
Zipper~\citep{huet:zipper} data-structure can be computed from the
derivative of signature functors. Interestingly, the derivative is
characterised by the existence of a universal arrow in the category
\(\PolyC{\LCCC{E}}\):
\begin{definition}[Differentiability~\citep{abbott:derivative}]

A polynomial \(F\) is differentiable in \(i\) if and only if, for any
polynomial \(G\), we have the following bijection of morphisms:
\[\begin{array}{l}
\Hom{\PolyC{\LCCC{E}}}{G \times \pi_i}{F} \\
\hline
\hline
\Hom{\PolyC{\LCCC{E}}}{G}{\Derive{i} F}
\end{array}
\]
Where \(\pi_i \triangleq \PolyDiag{\CatId}{I}{I}{k_i}{I}{\CatId}{I}\).
We denote \(\PolyDelta{\LCCC{E}}{i}\) the class of polynomials
differentiable in \(i\).

\end{definition}


\begin{proposition}

  Let \(F\) and \(G\) two polynomials in \(\PolyDelta{\LCCC{E}}{i}\). 

  If \(F\) ornaments \(G\), then \(\Derive{i} F\) ornaments
  \(\Derive{i} G\).

\end{proposition}


\begin{proof}


The proof simply follows from the functoriality of \(\Derive{i}\) over
\(\PolyDelta{\LCCC{E}}{i}\)~\citep[Section 6.4]{abbott:PhD}. In a
nutshell, this follows from the existence of the following cartesian
morphism:
\[ \Derive{i} F \times \pi_i \to^c F \to^c G \]
where the first component is the unit of the universal arrow while the
second component is the ornament from \(F\) to \(G\). By definition of
differentiability, we therefore have the desired cartesian morphism:
\[
\Derive{i} F \to^c \Derive{i} G
\]

\end{proof}


\begin{example}[Ornamentation of derivative]

Let us consider binary trees, with signature functor \(1 + A \times
X^2\). Balanced binary trees are an ornamentation of binary trees
(Example~\ref{ex:bbin-orn-bin-tree}). By the theorem above, we have
that the derivative of balanced binary trees is an ornament of the
derivative of binary trees.

\end{example}

The derivative is thus an example of an operation on datatypes that
preserves ornamentation. Knowing that the derivative of an ornamented
datatype is an ornamentation of the derivative of the original
datatype, we get that the order in which we ornament or derive a
datatype does not matter. This let us relate datatypes across such
transformations, thus preserving the structural link between them.


\section{Related work}


Ornaments were initially introduced by the second
author~\citep{mcbride:ornament} as a programming artefact. They were
presented in type theory, with a strong emphasis on their
computational contribution. Ornaments were thus introduced through a
universe. Constructions on ornaments -- such as the ornamental
algebra, algebraic ornament, and reornament -- were introduced as
programs in this type theory, relying crucially on the concreteness of
the universe-based presentation.

While this approach has many pedagogical benefits, it was also clear
that more abstract principles were at play. For example, in a
subsequent paper~\citep{dagand:fun-orn}, the authors successfully
adapted the notion of ornaments to another universe of inductive
families, whilst \citet{ko:modularising-inductive} explore datatype
engineering with ornaments in yet a third. The present paper gives
such an abstract treatment. This focus on the theory behind ornaments
thus complements the original, computational treatment.

Building upon that original paper, our colleagues
\citet{ko:modularising-inductive} also identify the pullback structure
-- called ``composition'' in their paper -- as significant, giving a
treatment for a concrete universe of ornaments and compelling examples
of its effectiveness for combining indexing disciplines. The
conceptual simplicity of our approach lets us subsume their type
theoretic construction as a mere pullback.


The notion of algebraic ornament was also treated categorically by
\citet{atkey:inductive-refinement}: instead of focusing on a
restricted class of functors, the authors described the refinement of
any functor by any algebra. The constructions are presented in the
generic framework of fibrations. The refinement construction described
in this paper, once specialised to polynomial functors, corresponds
exactly to the notion of algebraic ornament, as we have shown.


\citet{fiore:gadts} also give a model of inductive families in terms
of polynomial functors. To do so, they give a translation of inductive
definitions down to polynomials. By working on the syntactic
representation of datatypes, their semantics is \emph{defined by} this
translation. In our system, we can actually prove that descriptions --
our language of datatypes -- are equivalent to polynomial functors.


Finally, it is an interesting coincidence that cartesian morphisms
should play such an important role in structuring ornaments. Indeed,
containers stem from the work on shapely
types\citep{jay:shapely-types}. In the shape framework, a few base
datatypes were provided (such as natural numbers) and all the other
datatypes were grown from these basic blocks by a pullback
construction, \ie an ornament. However, this framework was simply
typed, hence no indexing was at play.



\section{Conclusion}


Our study of ornaments began with the equivalence between our universe
of descriptions and polynomial functors. This result lets us step away
from type theory, and gives access to the abstract machinery provided
by polynomials.  For practical reasons, the type theoretic definition
of our universe is very likely to change. However, whichever concrete
definition we choose will always be a syntax for polynomial
functors. We thus get access to a stable source of mathematical
results that informs our software constructions.


We then gave a categorical presentation of ornaments. Doing so, we get
to the essence of ornaments: ornamenting a datatype consists in
extending it with new information, and refining its indices. Formally,
this characterisation turns into a presentation of ornaments as
cartesian morphisms of polynomials.


Finally, we reported some initial results based on our explorations of
this categorical structure. We have translated the type theoretic
ornamental toolkit to the categorical framework. Doing so, we have
gained a deeper understanding of the original definitions. Then, we
have expressed the categorical definition of \(\PolyC{\LCCC{E}}\) in
terms of ornaments, discovering new constructions -- identity,
vertical, and horizontal composition -- in the process. Also, we have
studied the structure of \(\PolyC{\LCCC{E}}\), obtaining the notion of
pullback of ornaments.


\paragraph{Future work}


We have barely scratched the surface of \(\PolyC{\LCCC{E}}\): a lot
remain unexplored. Pursuing this exploration might lead to novel and
interesting ornamental constructions.
Also, our definition of ornaments in terms of polynomials might be
limiting. One can wonder if a more abstract criterion could be found
for a larger class of functors. For instance, the functor
\(\TypeAnn{\_\: \Terminal[\Cat{C}]}{[\Cat{C},\Cat{D}] \to \Cat{D}}\)
is a fibration for \(\Cat{D}\) pullback complete and \(\Cat{C}\)
equipped with a terminal object \(\Terminal[\Cat{C}]\). Specialised to
the categories of slices of \(\LCCC{E}\), the cartesian morphisms are
exactly our ornaments. What about the general case?

Finally, there has been much work recently on homotopy inductive
types~\citep{awodey:inductive-homotopy}. Coincidentally, the formalism
used in these works is based on W-types, \ie the type theoretic
incarnation of polynomial functors. It would be there be interesting
to study what ornaments could express in this framework.

\paragraph{Acknowledgements} 
We would like to thank Gabor Greif and Jos\'{e} Pedro Magalh\~{a}es
for their input on this paper. We also thank our colleagues
Cl\'{e}ment Fumex and Lorenzo Malatesta for their feedback on our
proofs. The authors are supported by the Engineering and Physical
Sciences Research Council, Grant EP/G034699/1.


\bibliographystyle{abbrvnat}
\bibliography{paper,%
  ../thesis-2011-phd/levitation,%
  ../report-2011-patch/paper,%
  ../report-2012-data/paper} 



\end{document}